\documentclass[11pt]{article}
\usepackage{authblk}
\usepackage{natbib}
\usepackage[margin=1.1in]{geometry}
\usepackage{amsmath}
\usepackage{amssymb}
\usepackage{amsthm}
\usepackage{enumerate}
\usepackage{enumitem}
\usepackage[usenames,dvipsnames,svgnames,table]{xcolor}
\usepackage{tikz}
\usepackage{tikz-cd}
\usepackage{array}
\usetikzlibrary{arrows,shapes,shapes.misc,positioning,calc}
\usepackage{hyperref}
\usepackage[capitalise]{cleveref}

\newtheorem{proposition}{Proposition}[section]
\newtheorem{theorem}[proposition]{Theorem}

\newtheorem{lemma}[proposition]{Lemma}

\theoremstyle{definition}

\newtheorem{remark}[proposition]{Remark}
\newtheorem{example}[proposition]{Example}

\newcommand*\circled[1]{\tikz[baseline=(char.base)]{\node[shape=circle,draw,inner sep=1.5pt] (char) {$#1$};}}

\newcommand*\mybox[1]{\tikz[baseline=(char.base)]{\node[shape=rectangle,draw,inner sep=1.5pt] (char) {$#1$};}}

\newcommand{\B}{\mathbb{B}}

\newcommand{\cs}{C} 
\newcommand{\compl}{\mathsf{c}}
\newcommand{\tf}{\tilde{f}}

\newcommand{\G}{\mathcal{G}}
\renewcommand{\k}{\kappa}
\renewcommand{\H}{\mathcal{H}}

\renewcommand{\P}{\mathcal{P}}
\renewcommand{\S}{\mathcal{S}}

\newcommand{\NN}[1]{\{1,\dots,{#1}\}}
\newcommand{\0}{\mathbf{0}}
\newcommand{\1}{\mathbf{1}}
\renewcommand\ss[1]{\substack{#1}}

\title{Boolean analysis of lateral inhibition}
\author{Elisa Tonello}
\author{Heike Siebert}
\affil{Discrete Biomathematics Group,\\Department of Mathematics and Computer Science,\\Freie Universit\"{a}t, Berlin, Germany}
\date{}

\begin{document}

\maketitle

\begin{abstract}
We study Boolean networks which are simple spatial models of the highly conserved Delta-Notch system. The models assume the inhibition of Delta in each cell by Notch in the same cell, and the activation of Notch in presence of Delta in surrounding cells. We consider fully asynchronous dynamics over undirected graphs representing the neighbour relation between cells. In this framework, one can show that all attractors are fixed points for the system, independently of the neighbour relation, for instance by using known properties of simplified versions of the models, where only one species per cell is defined. The fixed points correspond to the so-called fine-grained ``patterns'' that emerge in discrete and continuous modelling of lateral inhibition. We study the reachability of fixed points, giving a characterisation of the trap spaces and the basins of attraction for both the full and the simplified models. In addition, we use a characterisation of the trap spaces to investigate the robustness of patterns to perturbations. The results of this qualitative analysis can complement and guide simulation-based approaches, and serve as a basis for the investigation of more complex mechanisms.

\end{abstract}

\section{Introduction}

Lateral inhibition is a signalling mechanism that can induce the differentiation of cells
in developing tissues~\cite{sternberg1993falling,collier1996pattern}.
Transmembrane receptors of the \emph{Notch} family,
and the product of the \emph{Delta} gene acting as ligand,
have been identified as possible actors in this spatial differentiation phenomenon.
In its simplest form, lateral signalling causes cells to experience two different types of fate,
a \emph{primary} and a \emph{secondary} fate, corresponding to low and high levels of Notch.
The stimulation of Notch by the ligand Delta from adjacent cells induces the cell to
assume the secondary fate; high Notch activity, on its part, causes inhibition of Delta,
which promotes the lateral differentiation to the primary fate.
The result of this feedback is the emergence of spatial patterns of cells
of primary and secondary type.

Several mathematical models have been proposed for the investigation of the Delta-Notch pattern-generating mechanism
(e.g.,~\cite{collier1996pattern,webb2004oscillations,goessler2011component}).
In~\cite{collier1996pattern}, the authors choose a spatially-discretised model,
with dynamics described by systems of differential equations.
Their analysis highlights in particular that, when the feedback between cells is strong enough,
patterns of alternating high and low levels of Notch emerge, that do not depend on specific forms
for the regulations of species production, and on the parameters.
It is therefore natural to investigate whether the basic principles underlying the Delta-Notch system can
be identified also in a purely qualitative, Boolean framework.
Discrete models can often capture ``rules'' that govern properties of larger classes of systems
(see for instance~\cite{thomas1990biological,thomas2001multistationarity,albert2003topology}).
In this work we consider simple Boolean models, where only two variables, representing Notch and Delta,
are defined in each cell. The level of Delta in a cell is uniquely determined by the level of Notch in the same cell,
whereas multiple formulations for the dependence of Notch on the levels of Delta in neighbour cells
can be considered. In this work we focus on the assumption that the presence of one neighbour cell
with high level of Delta is sufficient for the activation of Notch.
In addition, we consider a simplified version of these models, where only one variable per cell
is defined, which inhibits variables in neighbouring cells.
The models we consider have already been analysed with computational approaches
for some specific network geometries~\cite{mendes2013composition,varela2018stable}.
Here we investigate properties that hold independently of the neighbour structure of the cells.

By considering the reduced, Boolean lateral inhibition models with one variable per cell,
one can use properties of threshold networks (\cite{goles1985decreasing}) to show that
all attractors for the asynchronous dynamics are fixed points.
These stable configurations or \emph{patterns} that emerge from the simple spatial interaction structure
we consider exhibit the same alternation of cells with low and high Notch level
observed in the ODE models of~\cite{collier1996pattern}.
The alternation requires each cell with low Notch to be surrounded by cells with high Notch,
and all cells with high Notch to have at least one neighbour with high Delta.
In other words, the Delta-Notch patterns are defined by the minimal vertex covers,
or maximal independent vertex sets, of the graph describing the neighbour relations (\cite{veliz2012computation}).
We ask which patterns can be reached under fully asynchronous dynamics from homogeneous initial conditions,
and show that all of them can be obtained (\cref{thm:homogeneous}).
We then provide a characterisation of the trap spaces
of the systems, that is, subspaces that the dynamics can not leave,
for both the two-variable and one-variable dynamics (\cref{thm:trapspaces,thm:trap-spaces-N}).
We give in addition a characterisation of the fixed points that are reachable from
a given initial condition, identifying some differences between the full and reduced models (\cref{thm:all-fixed,thm:reach-N}).
Determining the trap spaces allows us to study how patterns respond to perturbations.
In particular, we show that, for the models we consider, changes can not propagate beyond cells at distance two (\cref{sec:robustness}).
The spatial interaction structure consisting of internal inhibition and neighbour activation
can be thought of as a core model for lateral inhibition,
and it is not straightforward to determine which of the properties
we present here are preserved in larger or more complex models.
We discuss a generalisation of the models and additional open questions in~\cref{sec:generalisation,sec:conclusion}.

\section{Background}

In this section we set some notations and give some basic definitions.
We write $\B$ for the set $\{0,1\}$.
For $a\in\B$, we write $\bar{a}$ for $1-a$,
and given $n\in\mathbb{N}$, $I\subseteq\NN{n}$ and $x\in\B^n$, we denote by $\bar{x}^I$
the element with $\bar{x}^I_i=1-x_i$ for $i\in I$, and $\bar{x}^I_i=x_i$ otherwise.
If $I$ consists of only one element $i$, then we write $\bar{x}^i$ for $\bar{x}^I$,
and if $I=\NN{n}$, we write $\bar{x}$ for $\bar{x}^I$.
In the examples, we will simplify the notation and denote elements of $\B^n$
as sequences of $0$s and $1$s (e.g, we will write $100011$ for $(1,0,0,0,1,1)$).
We will also write $\0$ and $\1$ for the elements of $\B^n$ with all components
equal to $0$ or $1$ respectively.

A Boolean network on $n$ variables, with $n\in\mathbb{N}$, is defined by a function $f\colon\B^n\to\B^n$.
The set $\B^n$ is also called the state space of the Boolean network.
The dynamical system given by the iteration of $f$ is called \emph{synchronous dynamics}.
In biological contexts, the \emph{asynchronous dynamics} or \emph{asynchronous state transition graph}
of a Boolean network is often the object of interest.
The asynchronous dynamics $AD_f$ of $f$ is defined as the graph with vertex set $\B^n$,
and edge set $\{(x,\bar{x}^i) | f_i(x)\neq x_i, i=1,\dots,n\}$.

The interaction graph $G_f$ of a Boolean network $f$ is the labelled multi-digraph with vertex set $\NN{n}$
and admitting an edge $(j,i)$ with sign $s\in\{-1,1\}$ if $s=(f_i(\bar{x}^j)-f_i(x))(\bar{x}^j_j-x_j)\neq 0$
for some $x\in\B^n$.

Given $x\in\B^n$ and $I\subseteq\NN{n}$, we write $x[I]=\{y\in\B^n\ | \ y_i=x_i \ \forall i\notin I\}$.
We call $x[I]$ a \emph{subspace} of $\B^n$.
In the examples, we denote a subspace $x[I]$ using $x$
and replacing the elements $x_i$ with $i\in I$ with the symbol ``$\star$''.
For instance, $001\star\star1$ will denote the subspace of $\B^6$ with $I=\{4,5\}$ and $x_1=x_2=0$, $x_3=x_6=1$.

A set $A\subseteq \B^n$ is called a \emph{trap set} for a Boolean network $f$
if, for all $x\in A$, if $y$ is a successor for $x$ in the asynchronous dynamics,
then $y\in A$.
A trap set that is also a subspace is called a \emph{trap space}.
For each state $x\in\B^n$ there exists a unique minimal (with respect to set inclusion) trap space containing $x$,
which we denote by $\k(x)$.
Minimal trap sets are called \emph{attractors} for the asynchronous dynamics.
If an attractor consists of a single state, it is called \emph{fixed point} or \emph{steady state},
otherwise it is called a \emph{cyclic attractor}.

Given an attractor $A$, the \emph{(weak) basin of attraction} of $A$
is the set of states $x\in\B^n$ such that there exists a path from $x$ to $A$ in the asynchronous dynamics.
The \emph{strong basin of attraction} of $A$ is the set of states in the basin of attraction of $A$
that do not belong to the basin of attraction of any other attractor $A'\neq A$.

The following result, which can be found in~\cite{naldi2009reduction,pauleve2012static}, relates
properties of Boolean maps to properties of maps with a smaller number of variables.
For simplicity it is stated for the elimination of the $n^{th}$ variable,
but generalises to the elimination of any variable.
\begin{theorem}\label{thm:reduction}
  Consider a map $f\colon\B^n\to\B^n$ and define $\tf\colon\B^{n-1}\to\B^{n-1}$ as
  $\tf_i(x)=f_i(x, f_n(x,0))$ for each $x\in\B^{n-1}$, $i=1,\dots,n-1$.
  If $G_f$ does not admit an edge from $n$ to itself, then:
  \begin{enumerate}[label=(\roman*)]
  \item $x\in\B^{n-1}$ is a fixed point for $\tf$ if and only if
    $(x, f_n(x,0))$ is a fixed point for $f$.
  \item If $AD_{\tf}$ has a path from $x$ to $y$, then
    $AD_f$ has a path from $(x, f_n(x,0))$ to $(y, f_n(y,0))$.
  \end{enumerate}
\end{theorem}

It will be useful to relate the trap spaces of the full and reduced systems.

\begin{proposition}\label{prop:reduction}
  In the setting of~\cref{thm:reduction}, denote by $\pi_{n-1}$ the projection on the first $n-1$ components.
  \begin{enumerate}[label=(\roman*)]
    \item If $A$ is a trap space for $f$, then $\pi_{n-1}(A)$ is a trap space for $\tf$.
    \item If $A$ is a trap space for $\tf$, then $A\times\{a\}$ is a trap space for $f$
      if and only if $f_n(x,0)=f_n(x,1)=a\in\B$ for all $x\in A$.
    \item If $x[I]$ is a trap space for $\tf$, then $A=x[I]\times\{0,1\}$ is a trap space for $f$
      if and only if $f_i(y,0)=f_i(y,1)$ for all $y\in x[I]$ and $i\in I^{\mathsf{c}}$.
  \end{enumerate}
\end{proposition}
\begin{proof}
  $(i)$ Take $x\in\pi_{n-1}(A)$ and $y$ successor for $x$ in $AD_{\tf}$.
  Since $f_n(x,0)=f_n(x,1)$, either $(x, f_n(x, 0))$ is in $A$ or
  there exists an $a\in\{0,1\}$ such that $(x, a)$ is in $A$, and $(x, f_n(x, 0))$ is
  a successor for $(x, a)$ in $AD_f$.
  By~\cref{thm:reduction} $(ii)$ there is a path from $(x,f_n(x,0))$ to $(y,f_n(y,0))$ in $AD_f$,
  and, since $A$ is a trap space, $y$ is in $\pi_{n-1}(A)$, and we are done.

  $(ii)$ Suppose that $f_n(x,0)=f_n(x,1)=a\in\B$ for all $x\in A$, and
  take $(x,a)\in A\times\{a\}$, and $(y,b)$ successor for $(x,a)$ in $AD_f$.
  Then since $f_n(x,a)=a$, we have $b=a$, and $f_i(x,a)\neq x_i$ for some $i<n$.
  Hence $\tf_i(x)=f_i(x,f_n(x,a))=f_i(x,a)\neq x_i$ and $y$ is a successor for $x$ in $AD_{\tf}$,
  and therefore is in $A$.
  The other direction is trivial.

  $(iii)$ Suppose that $f_i(y,0)=f_i(y,1)$ for all $y\in x[I]$ and $i\in I^{\mathsf{c}}$, and
  take $(y,v)\in x[I]\times\{0,1\}$, and $(z,w)$ successor for $(y,v)$ in $AD_f$.
  If $z=y$, or $z=\bar{y}^i$ with $i\in I$, then clearly the successor is in $A$.
  If $z=\bar{y}^i$ with $i\in I^{\mathsf{c}}$, then $\tf_i(y)=f_i(y,f_n(y,0))=f_i(y,v)\neq y_i$,
  hence $z=\bar{y}^i$ is in $x[I]$, which concludes.
  The other direction is trivial.
\end{proof}

\subsection{A Boolean Delta-Notch model}\label{def:dn}

In this work we are interested in some Boolean networks that can be interpreted as
arising from the combination of multiple instances of a given Boolean function.
This approach is formalised for instance in~\cite{mendes2013composition,varela2018stable} and called composition of logical modules.
Here we use a different definition that can be recast in terms of compositions of modules.

We fix $L\in\mathbb{N}$ and consider an undirected connected graph $\G$ with vertex set $C=\{1,\dots,L\}$ and without loops.
We call the vertices \emph{cells} and $\G$ the \emph{cell graph} underlying the system,
as it represents a network of $L$ cells with some neighbouring relation.
For each $i\in\cs$, we write $\S(i)=\{j\in\cs\ | \ (i,j)\text{ edge in }\G\}$.
If $(i,j)$ is an edge in $\G$, we say that $i$ and $j$ are \emph{neighbours}.
In the examples we will consider for instance the \emph{path graph} or \emph{linear graph} $\P_L$,
the graph with vertices $\{1,\dots,L\}$ and edge set $\{(i,i+1)\ |\ i=1,\dots,L-1\}$,
representing a linear array of cells, where each internal cell has two neighbour cells ($\S(i)=\{i-1,i+1\}$),
and the first and last cell admit only one neighbour ($\S(1)=\{2\}$ and $\S(L)=\{L-1\}$).

The system in each cell is described by some Boolean variables, whose behaviour can depend on
the variables in the same cell or in neighbouring cells.
\cite{mendes2013composition,varela2018stable} also distinguish between \emph{input components}
and \emph{internal components}, the former being variables that can only depend on variables
in neighbouring cells, and the latter being variables that can only depend on other
variables from the same cell.
For the system studied in this work, we consider only two Boolean variables in each cell,
or one Boolean variable in each cell for the reduced models (see~\cref{sec:dn_reduced}).
We therefore do not introduce a general notation, but rather focus on special systems
with $2L$ or $L$ variables.

Given a cell graph $\G$, for each cell $i$ we consider a variable \emph{Notch} and a variable \emph{Delta},
that we denote $n_i$ and $d_i$, respectively, with $i=1,\dots,L$.
The space we consider is therefore $\B^{2L}$, and the network we study is a function $F\colon\B^{2L}\to\B^{2L}$.
Sometimes it will be convenient to denote an element $x\in\B^{2L}$ as
$x=(n,d)=(n_1,\dots,n_L,d_1,\dots,d_L)$, so that $x_i=n_i$ and $x_{i+L}=d_i$ for $i=1,\dots,L$.
Given $J\subseteq\cs$, we will write $J+L$ for the set $\{i+L\ | i\in J\}$,
and $J^{\compl}$ for $\cs\setminus J$.
For $I\subseteq\NN{2L}$ we define $I_N=I\cap\cs$, $I_D=\{i-L\ | \ i \in I\cap(\cs+L)\}$
and $\S(I)=\bigcup_{i\in I_N\cup I_D}\S(i)$.

In the simple model we consider, in each cell, Notch inhibits the production of Delta, with no other interaction taking place.
The logical function that encodes the regulation of Delta in cell $i$ is therefore
defined by $(n,d)\mapsto \bar{n}_i$.
Notch instead is activated by the presence of Delta in neighbouring cells.
Here we consider the following two possibilities:
either the presence of Delta in any of the neighbouring cells is sufficient for the activation of Notch,
or the presence of Delta in all of the neighbouring cells is required.
This leads to the definition of two possible Boolean functions for component $i$,
that we denote $F^{\wedge}$ and $F^{\vee}$ respectively:
\begin{equation*}
  F^{\wedge}_{i}(n,d)=\bigwedge_{j\in \S(i)} d_j, \hspace{20pt} F^{\vee}_{i}(n,d)=\bigvee_{j\in \S(i)} d_j.
\end{equation*}
Note however that $F^{\wedge}$ and $F^{\vee}$ verify
\begin{equation*}
  \overline{F^{\wedge}(\bar{n},\bar{d})} = \left(\overline{\bigwedge_{j\in \S(1)} \bar{d}_j}, \dots, \overline{\bigwedge_{j\in \S(L)} \bar{d}_j}, \bar{n}_1,\dots,\bar{n}_L\right)= F^{\vee}(n,d),
\end{equation*}
i.e., $F^{\wedge}$ and $F^{\vee}$ are conjugated under the function $x\mapsto\bar{x}$, and hence admit isomorphic
asynchronous state transition graphs.
It is therefore sufficient to limit our analysis to the function $F=F^{\vee}$.
We call $F$ a \emph{Boolean Delta-Notch system} over the graph $\G$.

\begin{example}\label{ex:L1}
  For $L=1$, we have $F(n_1,d_1)=(0,\bar{n}_1)$, and the system has only one attractor,
  the fixed point $01$, i.e., the dynamics converges to the state with low Notch and high Delta.
  The trap spaces for the system are $\star\star$, $0\star$ and $01$, and $\star\star$
  is the basin of attraction of $01$.
\end{example}

\begin{example}\label{ex:L2}
  For $L=2$, we find $F(n_1,n_2,d_1,d_2)=(d_2,d_1,\bar{n}_1,\bar{n}_2)$.
  The asynchronous dynamics, represented in~\cref{fig:L2}, admits two fixed points, $0110$ and $1001$,
  and two source states, $0101$ and $1010$. The remaining states are part of the same
  strongly connected component.
  Hence the trap spaces are given by the full state space and the two fixed points.
  The sets $\B^{4}\setminus\{1001\}$ and $\B^{4}\setminus\{0110\}$ are the basins of attraction
  of $0110$ and $1001$ respectively. There are no elements in the strong basin
  of attraction of $0110$ and $1001$, other than the fixed point itself.
  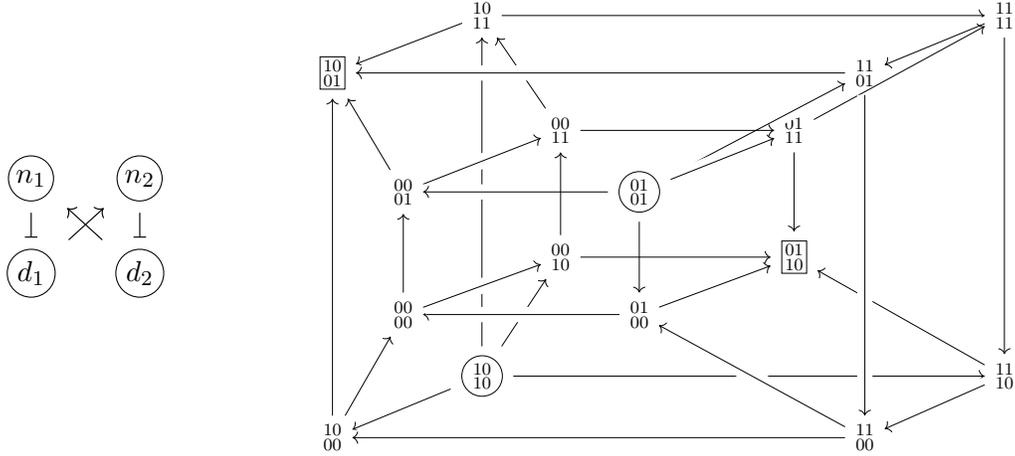
\begin{figure}
    \centering
    \begin{tikzcd}[row sep=small,column sep=small]
        \circled{n_1} \arrow[d,-|] & \circled{n_2} \arrow[d,-|] \\
        \circled{d_1} \arrow[ur] & \circled{d_2} \arrow[ul]
    \end{tikzcd}
      \hspace{40pt}
      \resizebox{0.6\linewidth}{!}{
    \begin{tikzcd}[ampersand replacement=\&,row sep=tiny,column sep=small]
           \&      \& \ss{10\\11} \&      \&      \& {\color{white}00} \&      \&      \& {\color{white}00} \& \ss{11\\11} \\
      \mybox{\ss{10\\01}}  \&      \&      \&      \&      \& {\color{white}00} \&      \& \ss{11\\01} \& {\color{white}00} \& \\
           \&      \&      \& \ss{00\\11} \&      \& {\color{white}00} \& \ss{01\\11} \&      \& {\color{white}00} \& \\
           \& \ss{00\\01} \&      \&      \& \circled{\ss{01\\01}} \& {\color{white}00} \&      \&      \& {\color{white}00} \& \\
           \&      \&      \& \ss{00\\10} \&      \& {\color{white}00} \& \mybox{\ss{01\\10}} \&      \& {\color{white}00} \& \\
           \& \ss{00\\00} \&      \&      \& \ss{01\\00} \& {\color{white}00} \&      \&      \& {\color{white}00} \& \\
           \&      \& \circled{\ss{10\\10}} \&      \&      \& {\color{white}00} \&      \&      \& {\color{white}00} \& \ss{11\\10} \\
      \ss{10\\00} \&      \&      \&      \&      \& {\color{white}00} \&      \& \ss{11\\00} \& {\color{white}00} \&
      \arrow[from=llllllluuuuuuu,to=uuuuuuu] 
      \arrow[from=llllllluuuuuuu,to=llllllllluuuuuu,crossing over] 
      \arrow[from=lllllllll,to=lllllllluu,crossing over] 
      \arrow[from=lllllllll,to=llllllllluuuuuu,crossing over] 
      \arrow[from=lllllllluu,to=lllllllluuuu,crossing over] 
      \arrow[from=lllllllu,to=llllllluuuuuuu] 
      \arrow[from=lllllllu,to=lllllllll,crossing over] 
      \arrow[from=lllllllu,to=u] 
      \arrow[from=lllllllu,to=lllllluuu] 
      \arrow[from=lllllllluuuu,to=llllllllluuuuuu,crossing over] 
      \arrow[from=lllllllluuuu,to=lllllluuuuu,crossing over] 
      \arrow[from=llllluuuu,to=lluuuuuu,crossing over] 
      \arrow[from=llllluuuu,to=lllllllluuuu,crossing over] 
      \arrow[from=llllluuuu,to=llluuuuu,crossing over] 
      \arrow[from=llllluuuu,to=llllluu,crossing over] 
      \arrow[from=llllluu,to=lllllllluu,crossing over] 
      \arrow[from=llllluu,to=llluuu,crossing over] 
      \arrow[from=uuuuuuu,to=lluuuuuu,crossing over] 
      \arrow[from=uuuuuuu,to=u] 
      \arrow[from=lllllluuu,to=llluuu] 
      \arrow[from=lllllluuu,to=lllllluuuuu] 
      \arrow[from=lllllluuuuu,to=llllllluuuuuuu] 
      \arrow[from=lllllluuuuu,to=llluuuuu] 
      \arrow[from=u,to=ll,crossing over] 
      \arrow[from=u,to=llluuu] 
      \arrow[from=llluuuuu,to=llluuu] 
      \arrow[from=llluuuuu,to=uuuuuuu] 
      \arrow[from=ll,to=lllllllll,crossing over] 
      \arrow[from=ll,to=llllluu,crossing over] 
      \arrow[from=lluuuuuu,to=ll,crossing over] 
      \arrow[from=lllllllluu,to=lllllluuu,crossing over] 
      \arrow[from=lluuuuuu,to=llllllllluuuuuu,crossing over] 
    \end{tikzcd}}
  \caption{Interaction graph and asynchronous state transition graph for a
      Boolean Delta-Notch model with $L=2$ (the levels of Delta are written below the corresponding levels of Notch).
      The fixed points are in rectangles. The circled states are source states.}\label{fig:L2}
  \end{figure}
\end{example}

\subsubsection{Model reduction}\label{sec:dn_reduced}
The model we described has $2L$ variables, none of which is autoregulated.
It will be convenient to work with the reduced network $N\colon\B^L\to\B^L$ obtained from $F$ by
elimination of the variables $d_1,\dots,d_L$ as delineated in~\cref{thm:reduction}.
For each $i=1,\dots,L$ we have
\begin{equation*}
  N_i(n) = \bigvee_{j\in\S(i)}\bar{n}_j = \overline{\bigwedge_{j\in\S(i)} n_j}.
\end{equation*}
By application of~\cref{thm:reduction}~$(i)$, the functions $F$ and $N$
have the same number of fixed points.
To a fixed point $n^*$ corresponds the fixed point $(n^*,\overline{n^*})$ for $F$.
In addition, from~\cref{thm:reduction}~$(ii)$, given $n, n'\in\B^L$, if there exists a path from
from $n$ to $n'$ in $AD_N$ then there exists a path from $(n,\bar{n})$ to $(n',\bar{n'})$ in $AD_F$.

\section{Asymptotic behaviour}\label{sec:asymptotic}

The asymptotic behaviour of Boolean Delta-Notch systems can be fully characterised.
By~\cref{thm:reduction}~$(i)$, the Boolean Delta-Notch system $F$
over a graph $\G$ has the same fixed points as the reduced network $N$.
The network $N$ is a normal OR-NOT network for its associated interaction graph, that is,
each component of $N$ is a disjunction, and its associated interaction graph has only negative edges.
The problem of finding fixed points of AND-OR networks and its relationship
to the problem of determining maximal independent sets or minimal vertex covers of a graph have been extensively investigated
(e.g.,~\cite{aracena2004fixed,veliz2012computation,aracena2014maximum,aracena2017fixed}).
As a corollary of~\cite[Proposition 3.5]{veliz2012computation},
the fixed points of $N$ are in one-to-one correspondence with the minimal (with respect to inclusion)
vertex covers of the graph $\G$.
A vertex cover of a graph is a subset $Q$ of the vertices of the graph such that
every edge of the graph has an endpoint in $Q$ (see for instance~\cite{west2001introduction}).

\begin{theorem}\label{thm:ss_vc}
  The fixed points of the Boolean Delta-Notch system over the graph $\G$
  are in one-to-one correspondence with the minimal vertex covers of the graph $\G$.
\end{theorem}

We refer to the fixed points also as stable spatial patterns, or simply patterns, for the system.
They are characterised by an alternating structure of primary fate and secondary fate cells,
which is determined by the structure of the cell graph $\G$.

\begin{remark}\label{rmk:atleast2fpts}
  It follows from~\cref{thm:ss_vc} that for any $i\in\cs$ there exists a fixed point $x$ for $N$
  that satisfies $x_i=0$, $x_j=1$ for all $j\in\S(i)$, and a fixed point $y$ for
  the Boolean Delta-Notch system over $\G$ that satisfies $y_i=\bar{y}_{i+L}=0$, $y_j=\bar{y}_{j+L}=1$ for all $j\in\S(i)$.
  In particular, if $L\geq 2$, then $N$ and $F$ admit at least two fixed points.
\end{remark}

A result on threshold networks can be used to show that $F$ and $N$ do not admit cyclic attractors.
A Boolean network $f\colon\B^n\to\B^n$ is called a (strict) \emph{threshold network} (\cite{goles1985decreasing})
if there exist a matrix $A\in\mathbb{R}^{n\times n}$ and a vector $b\in\mathbb{R}^n$ such that,
for all $i\in\NN{n}$, $f_i(x)=1$ if and only if $(Ax)_i>b_i$ and $f_i(x)=0$ if and only if $(Ax)_i<b_i$.

The network $N$ is a threshold network, with $A\in\{0,-1\}^{L\times L}$ and $b\in\mathbb{R}^L$ defined as follows:
\begin{equation*}
  \text{for all } i,j\in\{1,\dots,n\}, \hspace{15pt} A_{ij}=\begin{cases}-1 & \text{ if }j\in\S(i),\\0&\text{ otherwise,}\end{cases}
  \hspace{20pt}
  b_i=-|\S(i)|+\frac{1}{2}.
\end{equation*}
The \emph{energy} function $E\colon\{0,1\}^n\to\mathbb{R}$ associated to $A$ and $b$ is defined as
\begin{equation*}\label{eq:energy}
  E(x) = -\frac{1}{2}x^\mathsf{T}A x + b^{\mathsf{T}}x.
\end{equation*}
The matrix $A$ is symmetric and its diagonal elements are non-negative.
Under these conditions, the energy is strictly decreasing along asynchronous trajectories:
if $\bar{x}^i$ is a successor for $x$ in $AD_f$, then
\begin{equation*}\label{eq:monotone-energy}
  \begin{aligned}
    E(\bar{x}^i)-E(x) & = -\sum_{j\neq i}A_{ij}x_j(\bar{x}_i-x_i)-\frac{1}{2}A_{ii}(\bar{x}_i^2-x_i^2)+b_i(\bar{x}_i-x_i) \\
    & = -(\bar{x}_i-x_i)\left(\sum_{j=1}^nA_{ij}x_j-b_i\right)-\frac{1}{2}A_{ii}(\bar{x}_i-x_i)^2
    \leq -(\bar{x}_i-x_i)\left(\sum_{j=1}^nA_{ij}x_j-b_i\right)<0.
  \end{aligned}
\end{equation*}
As a consequence, the graph $AD_N$ does not admit any cyclic path.
This is a particular case of Proposition 1 in~\cite{goles1985decreasing},
which gives the following corollaries.

\begin{theorem}\label{thm:fpts-only-N}
  For each non-fixed point $x\in\B^{L}$ for a reduced Boolean Delta-Notch system $N$,
  there is a path in $AD_N$ from $x$ to a fixed point.
\end{theorem}

\begin{theorem}\label{thm:fpts_only}
  For each non-fixed point $x\in\B^{2L}$ for a Boolean Delta-Notch system $F$,
  there is a path in $AD_F$ from $x$ to a fixed point.
\end{theorem}
\begin{proof}
  Consider $(n,d)\in\B^{2L}$. Since there exists a path from $(n,d)$ to $(n,\bar{n})$,
  the conclusion follows from~\cref{thm:reduction}~$(ii)$ and~\cref{thm:fpts-only-N}.
\end{proof}

As a consequence, the asynchronous state transition graph of a Boolean Delta-Notch system does not admit cyclic attractors.
However, we will see that, unlike $AD_N$, the graph $AD_F$ contains cyclic paths (\cref{prop:path01tox_pathyto01_cycle}).

Observe that not every fixed point is reachable from every non-fixed point:
for instance, for the Boolean Delta Notch system over the path graph with $3$ nodes $\P_3$ there is no path
from $011100$ to the fixed point $101010$.
In the next section, we study the basins of attraction for both the one-variable and the two-variable models.

\section{Reachability of fixed points}\label{sec:reachability}

In the following, we consider the problem of determining which patterns can be obtained from some initial states.
The reachability of fixed points for Boolean Delta-Notch systems over hexagonal grids
from given initial conditions has been previously studied in~\cite{mendes2013composition}.
We start the section by showing that all the fixed points can be reached from homogeneous states,
that is, states where the levels are the same in every cell,
and identify other classes of states for which this property holds.

\subsection{Homogeneous initial conditions}

We first look at the reachability from homogeneous initial conditions for $N$.

\begin{theorem}\label{thm:homogeneous-N}
  Each fixed point $x\in\B^{L}$ is reachable in $AD_N$ from $\1$ and $\0$.
\end{theorem}
\begin{proof}
  We show that, for each fixed point $x\in\B^L$ for $N$, there is a path from $\1\in\B^L$ to $x$ in $AD_N$.
  The proof for $\0$ is similar.
  Consider a fixed point $x$ for $N$, and define $I(x)=\{i\in\cs\ | \ x_i=0\}$, $k=|I(x)|$.
  Set $x^0=\1$, choose an order $i_1,\dots,i_k$ for the indices in $I(x)$, and, for each $h=1,\dots,k$,
  define the state $x^{h}=\bar{\1}^{\{i_1,\dots,i_h\}}$.
  Then, for each $h=0,\dots,k-1$, $x^h_{i_{h+1}}=1$, $x_{i_{h+1}}=0$, and, since $x$ is fixed,
  for all $j\in\S(i_{h+1})$ we have $x_j=1$, so that $x^h_j=1$ and $N_{i_{h+1}}(x^h)=0$.
  Hence the asynchronous dynamics $AD_N$ admits an edge from $x^{h}$ to $x^{h+1}$, for $h=0,\dots,k-1$.
  In other words, there is a path in $AD_N$ from $x^{0}=\1$ to $x^{k}=x$.
\end{proof}

\begin{remark}\label{rmk:subspaces_transitions}
From each state $(n, d)$, there is a path to $(\bigvee_{j\in\S(1)}d_j,\dots,\bigvee_{j\in\S(L)}d_j,d)$
and a path to $(n,\bar{n})$ in $AD_F$.
Hence
\begin{itemize}
  \item if a state is reachable from $(\0,\0)$, it is reachable from $(n,\0)$ for all $n\in\B^L$;
  \item if a state is reachable from $(\1,\0)$, it is reachable from $(\1,d)$ for all $d\in\B^L$;
  \item for $L\geq 2$, if a state is reachable from $(\1,\1)$, it is reachable from $(n,\1)$ for all $n\in\B^L$;
  \item if a state is reachable from $(\0,\1)$, it is reachable from $(\0,d)$ for all $d\in\B^L$.
\end{itemize}
The asynchronous dynamics of every Boolean Delta-Notch system with $L\geq 2$ admits therefore
a cycle that includes all homogeneous states (see~\cref{fig:subspaces_transitions}, left).
In addition, the following result shows that all fixed points are reachable
from homogeneous states (see~\cref{fig:subspaces_transitions}, right, for an example).
\begin{figure}
  \centering
  \begin{minipage}{4.2cm}
  \begin{tikzpicture}[rounded corners]
    \pgfsetfillopacity{0.6}
    \fill[gray!35] (-0.6,-0.6) rectangle (0.6,2.6) [draw=gray!55];
    \fill[gray!35] (-0.5,-0.5) rectangle (2.5,0.5) [draw=gray!55];
    \fill[gray!35] (+1.4,-0.6) rectangle (2.6,2.6) [draw=gray!55];
    \fill[gray!35] (-0.5,+1.5) rectangle (2.5,2.5) [draw=gray!55];
    \node at (0,0) [shape=circle,fill=gray!85,inner sep=7] {};
    \node at (0,2) [shape=circle,fill=gray!85,inner sep=7] {};
    \node at (2,0) [shape=circle,fill=gray!85,inner sep=7] {};
    \node at (2,2) [shape=circle,fill=gray!85,inner sep=7] {};
    \node at (0,0) {$\0\0$};
    \node at (0,1) {$\0\star$};
    \node at (0,2) {$\0\1$};
    \node at (1,0) {$\star\0$};
    \node at (2,0) {$\1\0$};
    \node at (2,1) {$\1\star$};
    \node at (2,2) {$\1\1$};
    \node at (1,2) {$\star\1$};
    \draw [->] (-0.6,1) .. controls (-1.2,1.6) and (-0.6,2) .. (-0.3,2);
    \draw [->] (1,2.5) .. controls (1.5,3.1) and (2,2.6) .. (2,2.35);
    \draw [->] (2.6,1) .. controls (3.2,0.6) and (2.6,0.0) .. (2.35,0);
    \draw [->] (1,-0.5) .. controls (0.4,-1.2) and (0,-0.6) .. (0,-0.35);
  \end{tikzpicture}
  \end{minipage}
  \begin{minipage}{11.0cm}
  \resizebox{11.0cm}{!}{
  \newcommand*{\bs}{0.9}
  \newcommand*{\hg}{0.6}
  \newcommand*{\rd}{2.0}
  \begin{tikzpicture}[rounded corners]
    \node at (\bs*0.4,-\hg*0.5) {\huge $n$:};
    \node at (\bs*0.4,-\hg*1.5) {\huge $d$:};
    \foreach \i in {1,...,4} {
      \draw[rounded corners] (\i*\bs,-\hg) rectangle (\i*\bs+\bs,0);
      \draw[rounded corners,fill=black!100] (\i*\bs,-2*\hg) rectangle (\i*\bs+\bs,-\hg);
    }

    \foreach \i in {1,...,4} {
      \draw[rounded corners] (\i*\bs+3*\bs+\rd,-\hg) rectangle (\i*\bs+\bs+3*\bs+\rd,0);
      \draw[rounded corners,fill=black!100] (\i*\bs+3*\bs+\rd,-2*\hg) rectangle (\i*\bs+\bs+3*\bs+\rd,-\hg);
    }
    \draw[rounded corners,fill=black!100] (\bs+3*\bs+\rd,-\hg) rectangle (\bs+\bs+3*\bs+\rd,0);

    \foreach \i in {1,...,4} {
      \draw[rounded corners] (\i*\bs+3*\bs+\rd,-4*\hg) rectangle (\i*\bs+\bs+3*\bs+\rd,-3*\hg);
      \draw[rounded corners,fill=black!100] (\i*\bs+3*\bs+\rd,-5*\hg) rectangle (\i*\bs+\bs+3*\bs+\rd,-4*\hg);
    }
    \draw[rounded corners,fill=black!100] (\bs+4*\bs+\rd,-4*\hg) rectangle (\bs+\bs+4*\bs+\rd,-3*\hg);

    \foreach \i in {1,...,4} {
      \draw[rounded corners] (\i*\bs+6*\bs+2*\rd,-\hg) rectangle (\i*\bs+\bs+6*\bs+2*\rd,0);
      \draw[rounded corners,fill=black!100] (\i*\bs+6*\bs+2*\rd,-2*\hg) rectangle (\i*\bs+\bs+6*\bs+2*\rd,-\hg);
    }
    \draw[rounded corners,fill=black!100] (\bs+6*\bs+2*\rd,-\hg) rectangle (\bs+\bs+6*\bs+2*\rd,0);
    \draw[rounded corners,fill=black!100] (\bs+9*\bs+2*\rd,-\hg) rectangle (\bs+\bs+9*\bs+2*\rd,0);

    \foreach \i in {1,...,4} {
      \draw[rounded corners] (\i*\bs+6*\bs+2*\rd,2*\hg) rectangle (\i*\bs+\bs+6*\bs+2*\rd,3*\hg);
      \draw[rounded corners,fill=black!100] (\i*\bs+6*\bs+2*\rd,1*\hg) rectangle (\i*\bs+\bs+6*\bs+2*\rd,2*\hg);
    }
    \draw[rounded corners,fill=black!100] (\bs+6*\bs+2*\rd,2*\hg) rectangle (\bs+\bs+6*\bs+2*\rd,3*\hg);
    \draw[rounded corners,fill=black!100] (\bs+8*\bs+2*\rd,2*\hg) rectangle (\bs+\bs+8*\bs+2*\rd,3*\hg);

    \foreach \i in {1,...,4} {
      \draw[rounded corners] (\i*\bs+6*\bs+2*\rd,-4*\hg) rectangle (\i*\bs+\bs+6*\bs+2*\rd,-3*\hg);
      \draw[rounded corners,fill=black!100] (\i*\bs+6*\bs+2*\rd,-5*\hg) rectangle (\i*\bs+\bs+6*\bs+2*\rd,-4*\hg);
    }
    \draw[rounded corners,fill=black!100] (\bs+7*\bs+2*\rd,-4*\hg) rectangle (\bs+\bs+7*\bs+2*\rd,-3*\hg);
    \draw[rounded corners,fill=black!100] (\bs+9*\bs+2*\rd,-4*\hg) rectangle (\bs+\bs+9*\bs+2*\rd,-3*\hg);

    \foreach \i in {1,...,4} {
      \draw[rounded corners] (\i *\bs+9*\bs+3*\rd,-\hg) rectangle (\i *\bs+\bs+9*\bs+3*\rd,0);
      \draw[rounded corners,fill=black!100] (\i *\bs+9*\bs+3*\rd,-2*\hg) rectangle (\i *\bs+\bs+9*\bs+3*\rd,-\hg);
    }
    \draw[rounded corners,fill=black!100] (\bs+9*\bs+3*\rd,-\hg) rectangle (\bs+\bs+9*\bs+3*\rd,0);
    \draw[rounded corners,fill=black!100] (\bs+12*\bs+3*\rd,-\hg) rectangle (\bs+\bs+12*\bs+3*\rd,0);
    \draw[rounded corners,fill=white] (\bs+9*\bs+3*\rd,-2*\hg) rectangle (\bs+\bs+9*\bs+3*\rd,-\hg);

    \foreach \i in {1,...,4} {
      \draw[rounded corners] (\i *\bs+9*\bs+3*\rd,2*\hg) rectangle (\i *\bs+\bs+9*\bs+3*\rd,3*\hg);
      \draw[rounded corners,fill=black!100] (\i *\bs+9*\bs+3*\rd,\hg) rectangle (\i *\bs+\bs+9*\bs+3*\rd,2*\hg);
    }
    \draw[rounded corners,fill=black!100] (\bs+9*\bs+3*\rd,2*\hg) rectangle (\bs+\bs+9*\bs+3*\rd,3*\hg);
    \draw[rounded corners,fill=black!100] (\bs+11*\bs+3*\rd,2*\hg) rectangle (\bs+\bs+11*\bs+3*\rd,3*\hg);
    \draw[rounded corners,fill=white] (\bs+9*\bs+3*\rd,\hg) rectangle (\bs+\bs+9*\bs+3*\rd,2*\hg);

    \foreach \i in {1,...,4} {
      \draw[rounded corners] (\i *\bs+9*\bs+3*\rd,-4*\hg) rectangle (\i *\bs+\bs+9*\bs+3*\rd,-3*\hg);
      \draw[rounded corners,fill=black!100] (\i *\bs+9*\bs+3*\rd,-5*\hg) rectangle (\i *\bs+\bs+9*\bs+3*\rd,-4*\hg);
    }
    \draw[rounded corners,fill=black!100] (\bs+10*\bs+3*\rd,-4*\hg) rectangle (\bs+\bs+10*\bs+3*\rd,-3*\hg);
    \draw[rounded corners,fill=black!100] (\bs+12*\bs+3*\rd,-4*\hg) rectangle (\bs+\bs+12*\bs+3*\rd,-3*\hg);
    \draw[rounded corners,fill=white] (\bs+10*\bs+3*\rd,-5*\hg) rectangle (\bs+\bs+10*\bs+3*\rd,-4*\hg);

    \foreach \i in {1,...,4} {
      \draw[rounded corners] (\i *\bs+12*\bs+4*\rd,-\hg) rectangle (\i *\bs+\bs+12*\bs+4*\rd,0);
      \draw[rounded corners,fill=black!100] (\i*\bs+12*\bs+4*\rd,-2*\hg) rectangle (\i *\bs+\bs+12*\bs+4*\rd,-\hg);
    }
    \draw[rounded corners,fill=black!100] (\bs+12*\bs+4*\rd,-\hg) rectangle (\bs+\bs+12*\bs+4*\rd,0);
    \draw[rounded corners,fill=black!100] (\bs+15*\bs+4*\rd,-\hg) rectangle (\bs+\bs+15*\bs+4*\rd,0);
    \draw[rounded corners,fill=white] (\bs+12*\bs+4*\rd,-2*\hg) rectangle (\bs+\bs+12*\bs+4*\rd,-\hg);
    \draw[rounded corners,fill=white] (\bs+15*\bs+4*\rd,-2*\hg) rectangle (\bs+\bs+15*\bs+4*\rd,-\hg);

    \foreach \i in {1,...,4} {
      \draw[rounded corners] (\i *\bs+12*\bs+4*\rd,2*\hg) rectangle (\i *\bs+\bs+12*\bs+4*\rd,3*\hg);
      \draw[rounded corners,fill=black!100] (\i*\bs+12*\bs+4*\rd,\hg) rectangle (\i *\bs+\bs+12*\bs+4*\rd,2*\hg);
    }
    \draw[rounded corners,fill=black!100] (\bs+12*\bs+4*\rd,2*\hg) rectangle (\bs+\bs+12*\bs+4*\rd,3*\hg);
    \draw[rounded corners,fill=black!100] (\bs+14*\bs+4*\rd,2*\hg) rectangle (\bs+\bs+14*\bs+4*\rd,3*\hg);
    \draw[rounded corners,fill=white] (\bs+12*\bs+4*\rd,\hg) rectangle (\bs+\bs+12*\bs+4*\rd,2*\hg);
    \draw[rounded corners,fill=white] (\bs+14*\bs+4*\rd,\hg) rectangle (\bs+\bs+14*\bs+4*\rd,2*\hg);

    \foreach \i in {1,...,4} {
      \draw[rounded corners] (\i *\bs+12*\bs+4*\rd,-4*\hg) rectangle (\i *\bs+\bs+12*\bs+4*\rd,-3*\hg);
      \draw[rounded corners,fill=black!100] (\i*\bs+12*\bs+4*\rd,-5*\hg) rectangle (\i *\bs+\bs+12*\bs+4*\rd,-4*\hg);
    }
    \draw[rounded corners,fill=black!100] (\bs+13*\bs+4*\rd,-4*\hg) rectangle (\bs+\bs+13*\bs+4*\rd,-3*\hg);
    \draw[rounded corners,fill=black!100] (\bs+15*\bs+4*\rd,-4*\hg) rectangle (\bs+\bs+15*\bs+4*\rd,-3*\hg);
    \draw[rounded corners,fill=white] (\bs+13*\bs+4*\rd,-5*\hg) rectangle (\bs+\bs+13*\bs+4*\rd,-4*\hg);
    \draw[rounded corners,fill=white] (\bs+15*\bs+4*\rd,-5*\hg) rectangle (\bs+\bs+15*\bs+4*\rd,-4*\hg);

    \foreach \j in {0,...,3} {
      \node at (5*\bs+3*\bs *\j+\rd *\j+0.3*\rd,-\hg) {\huge $\rightarrow$};}
    \foreach \j in {2,3} {
      \node at (5*\bs+3*\bs *\j+\rd *\j+0.3*\rd,2*\hg) {\huge $\rightarrow$};}
    \foreach \j in {1,2,3} {
      \node at (5*\bs+3*\bs *\j+\rd *\j+0.3*\rd,-4*\hg) {\huge $\rightarrow$};}
    \node at (5*\bs+3*\bs+\rd+0.3*\rd,\hg) {\huge $\nearrow$};
    \node at (5*\bs+0.3*\rd,-3*\hg) {\huge $\searrow$};
  \end{tikzpicture}}
  \end{minipage}
  \caption{On the left, schematics of some transitions in the asynchronous state transition graph of a Delta-Notch system with $L\geq 2$.
    Homogeneous states are part of the same strongly connected component (\cref{rmk:subspaces_transitions}).
    On the right, some paths in the asynchronous dynamics associated to the graph $\P_4$, from
    the homogeneous state $(\1,\0)$ to the three fixed points (see~\cref{thm:homogeneous}).
    White represents high levels.}\label{fig:subspaces_transitions}
\end{figure}
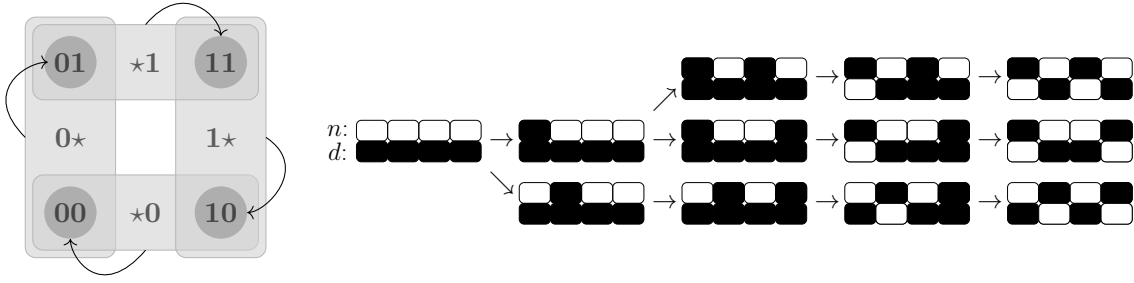
\end{remark}

\begin{theorem}\label{thm:homogeneous}
  Each fixed point $x\in\B^{2L}$ is reachable in $AD_F$ from any state in $\0\star\cup\star\0\cup\1\star\cup\star\1$.
\end{theorem}
\begin{proof}
  By~\cref{thm:homogeneous-N} and~\cref{thm:reduction}~$(ii)$, for each fixed point $(x,\bar{x})$ of $F$
  there is a path from $(\1,\0)$ to $(x,\bar{x})$. \cref{rmk:subspaces_transitions} then allows to conclude.
\end{proof}

\subsection{Trap spaces}

In this section, we give a characterisation of the trap spaces of Boolean Delta-Notch systems
and their reduced versions.

\begin{theorem}\label{thm:trap-spaces-N}
The trap spaces for $N$ are of the form $x[I]$, with $x$ fixed point, and for all
$i\in\S(I)\cap I^{\mathsf{c}}$ there exist $j\in\S(i)\cap I^{\mathsf{c}}$ such that $x_j=0$.
\end{theorem}
\begin{proof}
  Consider a subspace $x[I]$ as in the statement, and take $y\in x[I]$.
  We need to show that all successors of $y$ in the asynchronous state transition graph are in $x[I]$,
  or, in other words, $N_i(y)=y_i$ for all $i\notin I$.

  If $i\notin I$ and $j\notin I$ for all $j\in\S(i)$, then
  $N_i(y)=\bigvee_{j\in\S(i)}\bar{y}_{j}=\bigvee_{j\in\S(i)}\bar{x}_{j}=x_i=y_i$.
  Consider now the case of $i\notin I$ and $I\cap\S(i)\neq\emptyset$.
  Then there exists $k\in\S(i)\cap I^{\mathsf{c}}$ such that $x_k=0$, therefore
  $N_i(y)=\bigvee_{j\in\S(i)}\bar{y}_{j}=1=N_i(x)=x_i=y_i$.

  Vice versa, consider a trap space $x[I]$.
  Since we must have $N_i(x)=x_i$ for all $i\notin I$, and all attractors of $N$ are fixed points (see~\cref{thm:fpts-only-N}),
  we can assume that $x$ is a fixed point.
  Consider $i\in\S(I)\cap I^{\mathsf{c}}$ and take $j\in I\cap\S(i)$.
  Then there exists a state $y\in x[I]$ with $y_{j}=0$, and therefore $x_i=N_i(x)=\bigvee_{k\in\S(i)}\bar{y}_{k}=1$.
  Now take a state $z\in x[I]$ with $z_{k}=1$ for all $k\in \S(i)\cap I$. Then $x_i=1=\bigvee_{k\in\S(i)}\bar{z}_{k}=\bigvee_{k\in\S(i)\cap I^{\mathsf{c}}} \bar{z}_{k}$.
  This means that there exists $k\in\S(i)\cap I^{\mathsf{c}}$ such that $x_{k}=0$, which concludes.
\end{proof}

The trap spaces for $N$ correspond therefore to areas of fixed Notch, with borders of high Notch
sustained by cells with fixed, low levels of Notch.

The following proposition allows to identify the minimal trap space containing a pattern and some
of its adjacent states in $\B^L$.

\begin{proposition}\label{prop:robustnessNany}
  Consider $x\in\B^L$ fixed point for $N$ and a set of indices $H\subseteq\cs$.
  Define
  \[H_0=\{i\in H\ | \ x_i=0\}, \qquad H_1=\{i\in H\ | \ x_i=1\},\]
  \[K=\{j\in\S(H_1)\cap H^\compl\ | \ x_j=0\},\]
  \[J=\{j\in\S(K\cup H_0)\cap H^\compl\ | \ x_h=1 \ \forall h\in\S(j), h\notin K\cup H_0\},\]
  \[I=H\cup K\cup J.\]
  Then $x[I]$ is the minimal trap space for $N$ containing $x[H]$.
\end{proposition}
\begin{proof}
  Start by observing that
  \begin{align}
    x_i & = 0 \text{ for all } i\in H_0\cup K, \label{eq:0}\\
    x_i & = 1 \text{ for all } i\in H_1\cup J. \label{eq:1}
  \end{align}

  To show that $x[I]$ is a trap space, taking $h\in\S(I)\cap I^{\mathsf{c}}$, we show that $\S(h)\cap I^{\mathsf{c}}$
  is non-empty and $x_k=0$ for some $k\in\S(h)\cap I^{\mathsf{c}}$ (see~\cref{thm:trap-spaces-N}).
  \begin{enumerate}
    \item $h\in\S(H_0)$: we have $x_h=1$ from~\cref{eq:0}.
          Since $h\notin J$, there exists $k\in\S(h)$ such that $x_k=0$, $k\notin K\cup H_0$.
          From~\cref{eq:1} we have $k\notin H_1\cup J$, and we are done.
    \item $h\in\S(H_1)$, $h\notin\S(H_0)$: since $h\notin K$, by definition of $K$ we have $x_h=1$.
          Since $h$ is not in $J$, there are two cases:
          \begin{itemize}
            \item $h$ is in $\S(K)$ and has a neighbour $k\notin K\cup H_0$ with $x_k=0$, and using~\cref{eq:1} we are done, or
            \item $h$ is not in $\S(K)$. In this case $h$ has a neighbour $k$ such that $x_k=0$,
                  and this neighbour can not be in $H_0$ or $K$, and using~\cref{eq:1} we conclude.
          \end{itemize}
    \item $h\in\S(K)$, $h\notin\S(H)$: we have $x_h=1$ from~\cref{eq:0}.
          Since $h\notin J$, there exists $k\notin K\cup H_0$ with $x_k=0$, and using~\cref{eq:1} we are done.
    \item $h\in\S(J)$, $h\notin \S(K)\cup\S(H)$: there exists $k\in J$ such that $h\in\S(k)$.
          By definition of $J$, since $h$ is a neighbour of $J$ that is not in $K$ or $H_0$, we have $x_h=1$.
          Then $x_j=0$ for some neighbour $j$ of $h$. Since $h\notin \S(K)\cup\S(H)$, we have $j\notin K\cup H$ as required,
          and we conclude again using~\cref{eq:1}.
  \end{enumerate}

  To prove that $x[I]$ is minimal, for each $i\in I\setminus H$, we show that there exists
  a path in $AD_N$ from a state $y\in x[H]$ to a state $z$ with $z_i\neq x_i$.
  Take $y\in x[H]$ such that $y_i=1-x_i$ for all $i\in H$.
  By definition of $K$, there is a path from $y$ to $\bar{y}^K$,
  hence the minimal trap space containing $x[H]$ contains $x[H\cup K]$.
  Take $z\in x[H\cup K]$ with $z_i=1$ for all $i\in H\cup K$.
  Then for each $j\in J$ the state $\bar{z}^j$ is a successor for $z$,
  which concludes the proof.
\end{proof}

We now consider the trap spaces for $F$.
We first show how a trap space for $F$ can be obtained from a trap space for $N$.

\begin{proposition}\label{prop:trap-spaces-reduced-full}
  The subspace $x[I]$ is a trap space for $N$ if and only if
  the subspace $(x,\bar{x})[I\cup(I+L)]$ is a trap space for $F$.
\end{proposition}
\begin{proof}
  If the subspace $(x,\bar{x})[I\cup(I+L)]$ is a trap space for $F$,
  then by~\cref{prop:reduction}~$(i)$ the projection $x[I]$ onto the first $L$ variables is a trap space for $N$.

  Vice versa, consider $x[I]$ trap space for $N$.
  Recall that $N$ is obtained from $F$ by elimination of the variables $i+L$, with $i\in\cs$,
  in the sense of~\cref{thm:reduction}.
  Call $F'$ the function obtained from $F$ by eliminating the variables $i+L$ with $i\in I$,
  so that $N$ can be obtained from $F'$ by eliminating the variables $i+L$ with $i\in I^{\compl}$.
  Denote by $\pi_{I^\compl}$ the projection on the variables in $I^\compl$.

  For each $i\in I^{\mathsf{c}}$, $y\in x[I]$ and $z\in\B^L$, we have $F_{i+L}(y,z)=\bar{x}_i$.
  Hence by applying~\cref{prop:reduction}~$(ii)$ to each variable in $I^\compl+L$
  we find that the subspace $(x,\pi_{I^\compl}(\bar{x}))[I]$ is a trap space for $F'$.

  Take $i\in I^{\compl}$ and $(y,z)\in\B^{2L}$ such that $(y,\pi_{I^\compl}(z))\in (x,\pi_{I^\compl}(\bar{x}))[I]$.
  If $\S(i)\cap I=\emptyset$, we have $F_i(y,z)=\bigvee_{j\in\S(i)}z_j=\bigvee_{j\in\S(i)\cap I^{\mathsf{c}}}z_j$,
  and if $\S(i)\cap I\neq\emptyset$ we have, using~\cref{thm:trap-spaces-N},
  $F_i(y,z)=\bigvee_{j\in\S(i)}z_j\geq \bigvee_{j\in\S(i)\cap I^{\mathsf{c}}}z_j=1$.
  That is, none of the variables in $I^\compl$ and $I^\compl+L$ depend on variables in $I+L$.
  Hence~\cref{prop:reduction}~$(iii)$ applies to each variable in $I+L$ and we conclude.
\end{proof}

\begin{theorem}\label{thm:trapspaces}
  Given $I=I_N\cup(I_D+L)$ with $I_N,I_D\subseteq\cs$,
  the subspace $x[I]$ is a trap space for $F$ if and only if
  the subspace $x[I_N\cup(I_N+L)]$ is a trap space for $F$, $I_N\subseteq I_D$ and
  \begin{enumerate}[label=(\roman*)]
     \item $\S(I_D\setminus I_N)\cap I_D=\emptyset$
       and $x_j=0$ for all $j\in I_D\setminus I_N$;
     \item for all $i\in\S(I_D\setminus I_N)$ there exists
           $j\in\S(i)\cap I^{\mathsf{c}}_D$ such that $x_j=0$.
  \end{enumerate}
\end{theorem}
\begin{proof}
  If $x[I]$ is a trap space for $F$, since all attractors of $F$ are fixed points (see~\cref{thm:fpts_only}),
  we can assume that $x$ is a fixed point and write $x=(n,\bar{n})$.
  Then by~\cref{prop:reduction}~$(i)$ the subspace $n[I_N]$ is a
  trap space for $N$, and by~\cref{prop:trap-spaces-reduced-full} $x[I_N\cup(I_N+L)]$ is a trap space for $F$.
  In addition, $I_N\subseteq I_D$ follows from the definition of $F$.

  To prove $(i)$, consider $j\in I_D\setminus I_N$, and take an element $y\in x[I]$ with $y_{j+L}=1$.
  Then there exists a path from $y$ to a state $z$ with
  $z_k=1$ and $z_{k+L}=0$ for all $k\in\S(j)$, and since $x[I]$ is a trap space, we have $z\in x[I]$.
  Since $j\notin I_N$, we must have $x_j=\bigvee_{k\in\S(j)}z_{k+L}=0$.
  This is possible only if $I_D\cap\S(j)=\emptyset$ and $x_{k+L}=0$ for all $k\in\S(j)$.

  To show that $(ii)$ holds, take $k\in I_D\setminus I_N$. By point $(i)$, $x_k=0$ and therefore $x_i=1$ for all $i\in\S(k)$.
  Since, again by point $(i)$, any $i\in\S(k)$ is in $I^\compl_D$, there must exist a neighbour $j$
  of $i$ in $I^\compl_D$ such that $x_{j+L}=1$, which proves $(ii)$.

  Consider a subspace $x[I]$ such that $x[I_N\cup(I_N+L)]$ is a trap space for $F$,
  $I_N\subseteq I_D$ and $(i)$ and $(ii)$ hold, and take $y\in x[I]$.
  We need to show that $F_i(y)=y_i$ for all $i\notin I$.
  If $i\notin I$ and $i>L$, then $F_i(y)=\bar{y}_{i-L}=\bar{x}_{i-L}=x_i=y_i$.
  Similarly, if $i\notin I$, $i\leq L$ and $j\notin I$ for all $j\in\S(i)$, then
  $F_i(y)=\bigvee_{j\in\S(i)}y_{j+L}=\bigvee_{j\in\S(i)}x_{j+L}=x_i=y_i$.

  Consider now the case of $i\notin I$, $i\leq L$ and $I\cap\S(i)\neq\emptyset$.
  If $i\in\S(I_D\setminus I_N)$, then~$(i)$ implies $i\notin I_D$,
  and~$(ii)$ gives the existence of $k\in\S(i)\cap I_D^{\mathsf{c}}$ such that $x_k=0$.
  If $i\in\S(I_N)$ and $i\notin\S(I_D\setminus I_N)$, then since $x[I_N\cup(I_N+L)]$ is a trap space for $F$,
  by~\cref{prop:trap-spaces-reduced-full} and~\cref{thm:trap-spaces-N} there exists $k\in\S(i)$, $k\in I^\compl_D$ such that $x_k=0$.
  In both cases $y_{k+L}=x_{k+L}=1$ and $F_i(y)=\bigvee_{j\in\S(i)}y_{j+L}=1=F_i(x)=x_i=y_i$.
\end{proof}

The theorem states that the trap spaces for $F$ are found by lifting the trap spaces for $N$,
and optionally removing some constraints on Delta in isolated cells with low Notch, if the neighbouring
cells with high Notch are still sustained by other cells with high Delta.
Examples of trap spaces for a hexagonal grid and for a linear graph are given in~\cref{fig:trapspaces}.

The smallest trap spaces that are not fixed points are therefore of the form $x[\{i+L\}]$ for some steady state $x$
and some $i\in\cs$ such that $x_i=0$ and, for all $j\in\S(i)$, there is an index $k\in\S(j)$, $k\neq i$ such that $x_{k+L}=1$.
The trap space $x[\{i+L\}]$ consists of the fixed point $x$ and the state $\bar{x}^{i+L}$.
Under the same hypothesis, the subspace $x[\{i,i+L\}]$ is also a trap space.

\begin{figure}
  \begin{minipage}{10.5cm}
  \centering
  \begin{tikzpicture}[x=7.5mm,y=4.34mm]
    \tikzset{
      grayhex/.style={
        regular polygon,
        regular polygon sides=6,
        minimum size=10mm,
        inner sep=0mm,
        outer sep=0mm,
        rotate=0,
        fill=gray!80!white,
      draw}}
    \tikzset{
      whitehex/.style={
        regular polygon,
        regular polygon sides=6,
        minimum size=10mm,
        inner sep=0mm,
        outer sep=0mm,
        rotate=0,
        fill=white,
      draw}}
    \tikzset{
      blackhex/.style={
        regular polygon,
        regular polygon sides=6,
        minimum size=10mm,
        inner sep=0mm,
        outer sep=0mm,
        rotate=0,
        fill=black!90,
      draw}}

    \foreach \i in {0,...,5} 
      \foreach \j in {0,...,5} {
        \node[grayhex] at (2*\i,2*\j) {};
        \node[grayhex] at (2*\i+1,2*\j+1) {};}

      \node[whitehex] at (0,2) {};
      \node[whitehex] at (0,4) {};
      \node[whitehex] at (1,1) {};
      \node[whitehex] at (1,5) {};
      \node[whitehex] at (3,1) {};
      \node[whitehex] at (4,0) {};
      \node[whitehex] at (5,1) {};
      \node[whitehex] at (5,3) {};
      \node[whitehex] at (6,4) {};
      \node[whitehex] at (6,6) {};
      \node[whitehex] at (5,7) {};
      \node[whitehex] at (4,4) {};
      \node[whitehex] at (4,6) {};
      \node[whitehex] at (3,3) {};
      \node[whitehex] at (3,7) {};
      \node[whitehex] at (2,4) {};
      \node[whitehex] at (2,6) {};
      \node[whitehex] at (2,2) {};

      \node[blackhex] at (1,3) {};
      \node[blackhex] at (3,5) {};
      \node[blackhex] at (4,2) {};
      \node[blackhex] at (5,5) {};

      \node[whitehex] at (8,10) {};
      \node[whitehex] at (8,8) {};
      \node[whitehex] at (9,11) {};
      \node[whitehex] at (9,7) {};
      \node[whitehex] at (9,5) {};
      \node[whitehex] at (9,3) {};
      \node[whitehex] at (10,10) {};
      \node[whitehex] at (10,8) {};
      \node[whitehex] at (10,6) {};
      \node[whitehex] at (10,2) {};
      \node[whitehex] at (11,5) {};
      \node[whitehex] at (11,3) {};

      \node[blackhex] at (9,9) {};
      \node[blackhex] at (10,4) {};
  \end{tikzpicture}
  \end{minipage}
  \begin{minipage}{4.5cm}
  \resizebox{\linewidth}{!}{
  \arraycolsep=0.5pt\def\arraystretch{0.5}
  \begin{tikzcd}[ampersand replacement=\&,column sep=0.2em,row sep=small]
    \& \& \begin{array}{ccc}\star&\star&\star\\\star&\star&\star\end{array} \arrow[dll,-] \arrow[d,-] \arrow[dddr,xshift=5pt,-] \& \\
    \begin{array}{ccc}\star&1&0\\\star&0&1\end{array} \arrow[d,-] \& \& \begin{array}{ccc}0&1&\star\\1&0&\star\end{array} \arrow[d,-] \\
    \begin{array}{ccc}0&1&0\\\star&0&1\end{array} \arrow[dr,-] \& \& \begin{array}{ccc}0&1&0\\1&0&\star\end{array} \arrow[dl,-] \\
    \& \begin{array}{ccc}0&1&0\\1&0&1\end{array} \&        \& \begin{array}{ccc}1&0&1\\0&1&0\end{array}
  \end{tikzcd}}
  \end{minipage}
  \caption{On the left, example of levels of Notch characterising a trap space in a hexagonal grid.
    Areas of fixed Notch have a border with high Notch (in white) and an inner border with at least one
    neighbouring cell with low Notch (in black) for each cell at the outer border.
    Cells in grey have an undefined level of Notch.
    On the right, Hasse diagram for the subset relation of the trap spaces for the
    Boolean Delta-Notch system associated to the graph $\P_3$
    (the levels of Delta are written below the corresponding levels of Notch).}\label{fig:trapspaces}
\end{figure}

\begin{remark}\label{rmk:maximalts}
  For $L\geq 2$, the maximal non-trivial trap spaces for $N$ and $F$ are of the form
  $x[I]$ and $(x,\bar{x})[I\cup(I+L)]$ respectively, with $I=\cs\setminus(\{i\} \cup\S(i))$,
  $x$ fixed point for $N$ and $x_i=0$.
\end{remark}

Consider a trap space for $N$. The variables that are not fixed in the trap space identify connected subgraphs of $\G$,
and the dynamics corresponding to each connected component is a separate Boolean Delta-Notch system.

\begin{remark}\label{rmk:trapspace2fixedpoints}
  Consider a trap space $x[I]$ for $N$, and the subgraph $\G_I$ obtained by removing all vertices outside $I$
  and all the incident edges.
  Call $\G_1,\dots,\G_k$ the connected components of this subgraph, with vertices $C_1,\dots,C_k$ respectively.
  Write $N^1,\dots,N^k$ for the reduced Boolean-Delta Notch models associated to $\G_1,\dots,\G_k$,
  and $\pi^1,\dots,\pi^k$ for the projections on the variables in $C_1,\dots,C_k$ respectively.

  Then if $y^1,\dots,y^k$ are fixed points for $N^1,\dots,N^k$ respectively, and $y\in x[I]$
  satisfies $\pi^h(y)=y^h$ for $h=1,\dots,k$, then $y$ is a fixed point for $N$.

  In particular, by~\cref{rmk:atleast2fpts} any trap space $x[I]$ for $N$ with $\S(I)\cap I\neq\emptyset$ contains at least two fixed points,
  and any trap space $(x,\bar{x})[J]$ for $F$ with $\S(J)\cap J\neq\emptyset$ contains at least two fixed points.
\end{remark}

We have the following corollary of~\cref{prop:robustnessNany} and~\cref{thm:trapspaces}.

\begin{proposition}\label{prop:robustnessFany}
  Consider $x\in\B^{2L}$ fixed point for $F$ and a set of indices $H\subseteq\cs$.
  Then $x[I\cup(I+L)]$ is the minimal trap space for $F$ containing $x[H\cup(H+L)]$,
  where $I$ is defined as in~\cref{prop:robustnessNany}.
\end{proposition}

\subsection{Basins of attraction}

We now want to characterise the fixed points that are reachable from a given state,
for the reduced and the full models.

It is easy to see that the reduction in the number of variables has consequences on the reachability properties,
and some configurations for Notch that are reachable from a given state $(n,d)$ in a full two-variable model
might not be reachable from the state $n$ in the corresponding reduced model.
For instance, for the graph $\P_4$, there is no path in $AD_N$ from $1001$ to the fixed point $0110$,
but there is a path in $AD_F$ from $10010110$ to the fixed point $01101001$.

The following results characterise the states that are reachable in $AD_N$
from a given initial condition.
Given $I\subseteq\cs$, we use the notation $\G_I$ for the subgraph of $\G$
with set of vertices $I$ and set of edges consisting of all edges of $\G$
with both endpoints in $I$.

\begin{proposition}\label{prop:not-reachable}
  Given $x\in\B^L$, consider a subset $I\subseteq\cs$ such that $\G_I$ is connected,
  $x_i=0$ for all $i\in I$ and $x_h=1$ for all $h\in\S(I)\cap I^\compl$.
  If $y\in\B^L$ is such that $y_i=1$ for all $i\in I$, then $y$ is not reachable from $x$ in $AD_N$.
\end{proposition}
\begin{proof}
  We proceed by induction on the size of $I$.

  If $I=\{i\}$ for some $i\in\cs$, then by~\cref{thm:trap-spaces-N}
  the subspace $x[(I\cup\S(I))^{\mathsf{c}}]$ is a trap space for $AD_N$ and $y$ can not be reached from $x$.

  Assume that the conclusion holds for all sets of size smaller or equal to $k$ and suppose that $|I|=k+1$.
  By definition, $N_i(x)=1$ for all $i\in I$, and $N_j(x)=1$ for all $j\in\S(I)\cap I^\compl$.
  Take a path starting from $x$ and $z$ the first state in the path such that $z_i=1$ for some $i\in I$.
  By definition of $z$, we must have $z_j=1$ for all $j\in\S(I)\cap I^\compl$.
  Then any subset $J$ of $I\setminus\{i\}$ defining a connected component of $\G$ satisfies
  $|J|\leq k$, $z_j=0$ for all $j\in J$ and $z_h=1$ for all $h\in\S(J)\cap J^\compl$,
  and we conclude, using the induction hypothesis, that $y$ can not be reached from $z$,
  and therefore from $x$.
\end{proof}

To give the full characterisation of the fixed points reachable from a given state we will use
the following lemma. It formalises the idea that, given a state $x$ and some indices $I$
connected by edges in $\G$ and such that $x_i=0$ for all $i\in I$, it is possible, in the
asynchronous dynamics of $N$, to keep an arbitrary component $i$ in $I$ fixed to zero
while changing all other levels in $I$ from zero to one.

\begin{lemma}\label{lemma:tree-to-1}
  Given $x\in\B^L$, consider a subset $I\subseteq\cs$ such that $\G_I$ is connected
  and $x_i=0$ for all $i\in I$.
  Then for any $i\in I$ and $J\subseteq I\setminus\{i\}$ there is a path in $AD_N$ from $x$ to $\bar{x}^J$.
\end{lemma}
\begin{proof}
  Fix $i\in I$ and $J\subseteq I\setminus\{i\}$. Since $\G_{I}$ is connected, there exists a spanning tree $T$ for $\G_{I}$ with $i$ as root vertex.
  Denote by $m$ the maximum distance of the vertices in $I$ from $i$ along the paths in $T$.
  For $k=0,\dots,m$, denote by $I_k$ the vertices in $I$ at distance $k$ from $i$ in $T$,
  define $J_k=J\cap(\bigcup_{j=m-k+1}^{m}I_j)$ and set $y^k=\bar{x}^{J_k}$.
  We thus have $y^0=x$, $y^m=\bar{x}^J$ and $y^k=\overline{y^{k-1}}^{J\cap I_{m-k+1}}$ for $k=1,\dots,m$.
  Then for each $k=1,\dots,m$ we have $y^{k-1}_j=0$ for $j\in I_{m-k}$ and $j\in I_{m-k+1}$, hence
  $N_j(y^{k-1})=\bigvee_{h\in\S(j)}y^{k-1}_h\geq\bigvee_{h\in\S(j)\cap I_{m-k}}y^{k-1}_h=1$ for all $j\in I_{m-k+1}$,
  and therefore $AD_N$ has a path from $y^{k-1}$ to $y^k$, which concludes.
\end{proof}

\begin{theorem}\label{thm:reach-N}
  Given $x\in\B^L$, consider the partition of $\{i\in\cs|x_i=0\}$ into
  maximal disjoint sets $(I_\nu)_\nu$ such that $\G_{I_\nu}$ is connected.
  A fixed point $y\in\B^L$ for $N$ is reachable from $x$ in $AD_N$ if and only if
  for each $I_\nu$ there exists $i\in I_\nu$ such that $y_i=0$.
\end{theorem}
\begin{proof}
  Suppose that, for some $I\in (I_\nu)_\nu$, $y_i=1$ for all $i\in I$.
  Observe that $x_h=1$ for all $h\in\S(I)\cap I^\compl$.
  Then the conclusion follows from~\cref{prop:not-reachable}.

  For the other direction, suppose that $y\in\B^L$ is a fixed point such that for each set $I_\nu$
  there exists $i\in I_\nu$ with $y_i=0$. Define $I^1_\nu=\{j\in I_\nu|y_j=1\}$.
  Observe that the sets $I^1_\nu$ are disjoint.
  By~\cref{lemma:tree-to-1}, for each $\nu$, there exists a path from $x$ to $\bar{x}^{I^1_\nu}$.
  Since the components in $I_\nu$ do not depend on components in $I_\mu$ for $\mu\neq \nu$,
  there exists a path from $x$ to a state $z$ with $z_j=1$ for each $j\in\cs$ such that $y_j=1$.

  Now take the set $I^0=\{i\in\cs\ | \ z_i=1, y_i=0\}$.
  Since $y$ is fixed, $y_j=1$, and hence $z_j=1$, for all $j\in\S(I^0)$.
  Hence there is a path from $z$ to $\bar{z}^{I^0}=y$, which concludes.
\end{proof}

We can use the result to characterise the strong basin of attraction of a fixed point.
This is given by the trap spaces containing the fixed point, such that
the cells corresponding to non-fixed variables are isolated.

\begin{proposition}\label{prop:strong-basins-N}
  For each fixed point $x\in\B^{L}$, the strong basin of attraction is given by the union
  of the trap spaces $x[I]$ with $I\neq\cs$ and $\S(I)\cap I=\emptyset$.
\end{proposition}
\begin{proof}
  For $L=1$, the result is trivial. For $L\geq 2$, first observe that, by~\cref{thm:trap-spaces-N},
  if $x[I]$ is a trap space with $I\neq\cs$ and $\S(I)\cap I=\emptyset$, then
  for all $i\in I$ and $j\in\S(i)$ we have $j\in I^{\mathsf{c}}$, $x_j=1$ and $x_i=0$,
  and $x[I]$ contains only the fixed point $x$.
  Hence $x[I]$ is contained in the strong basin of attraction of $x$.
  It remains to show that any other state in the basin of attraction of $x$ is also in the basin of
  attraction of some other fixed point.

  Consider a state $z$ in the basin of attraction of $x$ that does not belong to a trap
  space of the form $x[I]$ with $I\neq\cs$ and $\S(I)\cap I=\emptyset$.
  Consider the partition of $\{i\in\cs|z_i=0\}$ into
  maximal disjoint sets $(I_\nu)_\nu$ such that $\G_{I_\nu}$ is connected, as in~\cref{thm:reach-N}.

  If $z_i=1$ for all $i\in\cs$, or $z_i=0$ for all $i\in\cs$, we conclude using~\cref{rmk:atleast2fpts} and~\cref{thm:homogeneous-N}.

  If $|I_{\nu}|=1$ for all $\nu$, by~\cref{thm:trap-spaces-N} the subspace $x[I]$
  with $I^{\mathsf{c}}=\cup_{\nu}I_{\nu}\cup\S(\cup_{\nu}I_{\nu})$ is a trap space containing $x$ and $z$,
  and $I\neq\cs$.
  Hence, by hypothesis, $\S(I)\cap I$ is non-empty, and by~\cref{rmk:trapspace2fixedpoints},
  $x[I]$ contains another fixed point $y$. In addition, by~\cref{thm:reach-N} $x_i=0$ for all $i\in\cup_{\nu} I_{\nu}$,
  and since $y$ coincides with $x$ outside $I$, $z$ and $y$ also verify the
  hypotheses of~\cref{thm:reach-N} and $y$ is reachable from $z$.

  Now suppose that, for some $\mu$, $I_\mu$ contains more than one index.
  By~\cref{thm:reach-N}, there exists $i\in I_\mu$ such that $x_i=0$. Take $j\in I_\mu$ with $j\in\S(i)$.
  Write $x[I]$ for the minimal trap space containing $x[\{i,j\}]$.
  By~\cref{prop:robustnessNany}, $I$ might contain cells at distance $1$ or $2$ from $\{i,j\}$,
  and cells $h$ at distance $2$ satisfy $x_h=1$.
  For any $\nu\neq\mu$, since $I_{\nu}\cap I_{\mu}=\emptyset$ and each $I_{\nu}$ is connected,
  we have that every index $h$ in $I_{\nu}\cap I$ is at distance $2$ from $\{i,j\}$,
  and hence satisfies $x_h=1$. Since $x$ is reachable from $z$,
  by~\cref{thm:reach-N} there must exists $h\in I_\nu$, $h\notin I$ such that $x_h=0$.
  By~\cref{rmk:trapspace2fixedpoints} there exists another fixed point $y\neq x$, $y\in x[I]$,
  that satisfies $y_j=0$.
  Since $y$ coincides with $x$ outside $I$, for any $\nu$ there exists $h\in I_\nu$ such that $y_h=0$,
  and by~\cref{thm:reach-N} the state $z$ is in the basin of attraction of both $x$ and $y$.
\end{proof}

We now move on to the two-variable models.
For the asynchronous dynamics associated to the network $F$, we show that all the attractors
found in the minimal trap space containing the state are reachable.
While in the reduced model any change in Notch immediately translates into a different behaviour of
the cell in terms of effects on the neighbouring cells,
in the full model the additional intermediate variables play a memory role
which allows for a delay in the effect, resulting in more possible asynchronous paths.
This different behaviour might be relevant in a biological context, where processes that take
place at different times scales are involved, for example including signalling and gene regulation mechanisms.
The effects generated by interacting processes with significantly different time scales
might be more faithfully captured by the extended models.

The idea of the proof of the lemma below is as follows.
If a given state $x$ does not belong to any non-trivial trap space,
a path can be exhibited from $x$ to a state with homogeneous, low levels of Delta.
The path can be obtained through the following steps: first all low levels of Delta that can increase are increased,
but only if they are not completely surrounded by cells with high Notch and low Delta.
Then, Notch levels are increased in all cells where it is possible.
Since $x$ does not belong to any non-trivial trap space, it is then sufficient to bring
all the levels of Delta down.

\begin{lemma}\label{lem:xtohomog}
  Consider $x\in\B^{2L}$ such that $\k(x)=\B^{2L}$.
  Then there exists a path in $AD_F$ from $x$ to $(\1,\0)$.
\end{lemma}
\begin{proof}
  It is sufficient to show that there exists a path in $AD_F$ from $x$ to a state $z$ with $z_{i+L}=0$ for all $i\in\cs$ (see~\cref{rmk:subspaces_transitions}).

  Define the set $J=\{i\in\cs\ | \ x_i=0 \text{ and } x_j=1, x_{j+L}=0 \text{ for all } j \in \S(i)\}$.
  If $x_{i+L}=1$ for some $i\in J$, then the subspace $y[I\cup(I+L)]$ with $I=\cs\setminus(\{i\}\cup\S(i))$
  satisfies the conditions of~\cref{thm:trapspaces} and is a trap space containing $x$.
  Since $x$ does not belong to any non-trivial subspace, we have $x_{i+L}=0$ for all $i\in J$.

  Consider the set of indices $J_1=\{i\in\cs\ | \ x_i=x_{i+L}=0\}$.
  Then $J\subseteq J_1$, and there is a path in $AD_F$ from $x$ to $v=\bar{x}^{(J_1+L)\setminus (J+L)}$.

  Now define $J_2=\{i\in\cs\ |\ v_i=0 \text{ and } v_{j+L}=1 \text{ for some } j\in\S(i)\}$.
  Again, there is a path in $AD_F$ from $v$ to $w=\bar{v}^{J_2}$.
  Note in addition that $w\geq v\geq x$, so that $x_i=1$ implies $w_i=1$.
  If $x_i=0$ for some $i\in\cs$, we have:
  \begin{itemize}
    \item If $i\in J$, $w_{i+L}=v_{i+L}=x_{i+L}=0$.
    \item If $i\notin J$ and $x_{j+L}=0$ for all $j\in\S(i)$, then there exists $k\in\S(i)$
      such that $x_k=0$ and $v_{k+L}=1$, so that $w_i=1$.
    \item If $i\notin J$ and there exists $k\in\S(i)$ such that $x_{k+L}=1$, then $v_{k+L}=1$ and $w_i=1$.
  \end{itemize}
  In summary, $w$ verifies $w_i=1$ for all $i\in\cs\setminus J$ and $w_{i+L}=0$ for $i\in J$.
  As a consequence, taking $J_3=\{i\in\cs\setminus J\ | \ w_i=w_{i+L}=1\}$, we have that the state $z=\bar{w}^{J_3+L}$
  is reachable from $w$ and verifies $z_{i+L}=0$ for all $i\in\cs$, and we conclude.
\end{proof}

The previous lemma shows that, from states that do not belong to any non-trivial subspace, any
homogeneous state can be reached. This result, combined with~\cref{thm:homogeneous}, gives that any fixed point
can be reached from such initial conditions.
When the initial state $y$ belongs to some non-trivial subspace, the fixed points that can be reached
are limited by the minimal subspace $\k(y)$ containing $y$.
To prove that all fixed points contained in $\k(y)$ can be reached from $y$, we consider the projection
of the dynamics on the subspace $\k(y)$, and study it as the combination of smaller Boolean Delta-Notch subnetworks.
It can be shown that, in general, in such a scenario, the full dynamics in the trap spaces
can be derived from the dynamics of the isolated active subnetworks (\cite{siebert2009deriving}).
Here we give a self-contained proof.

\begin{proposition}\label{prop:path01tox_pathyto01_cycle}
  Consider a fixed point $x$ and a trap space $x[I]$ for $F$ with $I_D\neq\cs$, and call $z$ the state in $x[I]$
  with $z_i=1$ for $i\in I$, $i\leq L$ and $z_i=0$ for $i\in I$, $i\geq L+1$. Then:
  \begin{enumerate}[label=(\roman*)]
    \item There exists a path in $AD_F$ from $z$ to $x$.
    \item There exists a path in $AD_F$ from any state $y\in x[I]$ with $\k(y)=x[I]$ to $z$.
    \item If $\S(I)\cap I_D=\emptyset$, then $x[I]$ contains exactly one fixed point.
    \item If $\S(I)\cap I_D\neq\emptyset$, then $x[I]$ contains at least two fixed points,
    and $AD_F$ admits a cycle with vertices in $x[I]$.
  \end{enumerate}
\end{proposition}
\begin{proof}
  Consider the subgraph $\G'$ of $\G$ obtained by removing all vertices outside $I_D$ and all the incident edges.
  Then $\G'$ can be decomposed into connected graphs $\G_1,\dots,\G_k$ with vertex sets $C_1,\dots,C_k$ respectively.
  We will now consider the projection of the dynamics on the components identified by $C_1,\dots,C_k$.
  For each $h\in\{1,\dots,k\}$, writing $C_h=\{j_1,\dots,j_{|C_h|}\}$, and denoting by
  $\pi_i\colon\B^{2L}\to\B$ the projection on the $i^{th}$ component,
  consider the maps $\pi^h\colon\B^{2L}\to\B^{2|C_h|}$ defined by $\pi^h=(\pi_{j_1},\pi_{j_2},\dots,\pi_{j_{|C_h|}},\pi_{j_1+L},\pi_{j_2+L},\dots,\pi_{j_{|C_h|+L}})$,
  and $\iota^h\colon\B^{2|C_h|}\to\B^{2L}$, $\iota^h_i(y)=y_i$ for $i\in C_h\cup (C_h+L)$, $\iota^h_i(y)=x_i$ for $i\notin C_h\cup (C_h+L)$.
  Define, for each $h\in\{1,\dots,k\}$, the Boolean network $F^h\colon\B^{2|C_h|}\to\B^{2|C_h|}$, $F^h=\pi^h\circ F\circ\iota^h$.
  Then, $(y,\bar{y}^i)$ is a transition in $AD_F$ for some $y\in x[I]$ and $i\in C_h$ if and only if
  $(\pi^h(y), \overline{\pi^h(y)}^i)$ is a transition in $AD_{F^h}$.
  In addition, $\pi^h(x)$ is a fixed point for $F^h$.

  Since, by~\cref{thm:trapspaces}~$(ii)$, $x_{j+L}=0$ for all $j\in\S(I)\cap I_D^{\compl}$,
  we have that, for each $h\in\NN{k}$, $i\in C_h$ and $y\in x[I]$, $F_i(y)=\bigvee_{j\in\S(i)}y_{j+L}=\bigvee_{j\in\S(i)\cap C_h}y_{j+L}$,
  that is, the dynamics on each connected component $C_h$ is not influenced by variables outside $C_h$,
  and $F^h$ is a Boolean Delta-Notch system on $\G_h$.
  Then $(i)$ follows from the application of~\cref{thm:homogeneous} to each Boolean network $F^h$.

  If $y\in x[I]$ satisfies $\k(y)=x[I]$, first observe that, if $i\in I_D$ and $i\notin I_N$,
  then by~\cref{thm:trapspaces}~$(i)$ $x_{i+L}=1$, $x_i=y_i=z_i=0$, and $y_{i+L}=z_{i+L}=0$.
  In addition, for each $h=1,\dots,k$, $\pi^h(y)$ does not belong to any non-trivial trap space defined by $F^h$.
  $(ii)$ is therefore a consequence of~\cref{lem:xtohomog}.

  To prove $(iii)$, consider $w$ fixed point in $x[I]$ and $i\in I$.
  Since by~\cref{thm:trapspaces}~$(ii)$ $x_{j+L}=w_{j+L}=0$ for all $j\in\S(i)$, we have $x_i=w_i=0$ and $x_{i+L}=w_{i+L}=1$, and hence $w=x$.

  The first part of $(iv)$ was shown in~\cref{rmk:trapspace2fixedpoints},
  and the second follows from~\cref{rmk:subspaces_transitions}.
\end{proof}

\begin{theorem}\label{thm:all-fixed}
  For every $y\in\B^{2L}$ and for every fixed point $x\in\k(y)$ there exists a path
  from $y$ to $x$ in $AD_F$.
\end{theorem}
\begin{proof}
  Take $y\in\B^{2L}$ and any $x$ fixed point in $\k(y)$.
  By Theorem~\ref{thm:trapspaces}, we can write $\k(y)=x[I]$ for some $I\subseteq\NN{2L}$.
  We conclude using~\cref{prop:path01tox_pathyto01_cycle}, $(ii)$ and~$(i)$.
\end{proof}
The theorem states that, for any Boolean Delta-Notch model and any state $y$,
all attractors that are contained in the minimal trap space containing $y$ are reachable from $y$.
As a corollary of the theorem, the basin of attraction of a fixed point $x$ is found
by taking all the trap spaces defined starting from $x$ as in~\cref{thm:trapspaces},
and removing all states found in trap spaces that do not contain the fixed point $x$.
We can reformulate the observation as follows.

\begin{proposition}\label{prop:weak_basins}
  For $L\geq 2$, for each fixed point $x\in\B^{2L}$, the basin of attraction is given by
  \[\B^{2L}\setminus \bigcup_{t \in M, x\notin t} t,\]
  where $M$ is the set of maximal, non-trivial trap spaces.
\end{proposition}
\begin{proof}
  Write $T$ for the set of all non-trivial trap spaces.
  Consider a fixed point $x$, and denote by $B$ its basin of attraction.
  Given $y\in B^{\mathsf{c}}$, by~\cref{thm:all-fixed} we have that $x\notin\k(y)$,
  hence the equality $B^{\mathsf{c}}=\bigcup_{t \in T, x\notin t} t$.
  It remains to show that any state $y$ contained in a trap space that does not contain $x$
  is also contained in a maximal trap space that does not contain $x$.
  Suppose that $y\in z[I]$ with $z$ fixed point and $x\notin z[I]$.
  Then there exist an $i\notin I$, $i\in\cs$ such that $z_i=0$ and $x_i=1$.
  The characterisation of trap spaces in~\cref{thm:trapspaces} implies that $\{i\}\cup\S(i)\subseteq I^{\mathsf{c}}$,
  and by~\cref{rmk:maximalts} the subspace $z[J\cup(J+L)]$ with $J=\cs\setminus(\{i\}\cup\S(i))$
  is a maximal non-trivial trap space that contains $y$ and does not contain $x$.
\end{proof}

We can also characterise the strong basins of attraction.

\begin{proposition}\label{prop:strong_basins}
  For each fixed point $x\in\B^{2L}$, the strong basin of attraction is given by the union
  of the trap spaces $x[I]$ with $I_D\neq\cs$ and $\S(I)\cap I=\emptyset$.
\end{proposition}
\begin{proof}
  For $L=1$, the result is trivial. For $L\geq 2$, first observe that, by~\cref{prop:path01tox_pathyto01_cycle}~$(iii)$,
  the trap spaces $x[I]$ with $I\neq\cs$ and $\S(i)\cap I=\emptyset$ for all $i\in I$ are contained in the strong basin of attraction of $x$.
  It remains to show that any other state in the basin of attraction of $x$ is also in the basin of
  attraction of some other fixed point.

  Consider a state $z$ in the basin of attraction of $x$, and suppose that the trap space $\k(z)$ can be written as $x[I]$
  with $I$ such that there exist $i,j\in I$ with $j\in\S(i)$.
  By~\cref{rmk:trapspace2fixedpoints} there exists another fixed point $y\neq x$, $y\in x[I]$.
  Then by~\cref{thm:all-fixed} the state $z$
  is in the basin of attraction of $x$ and in the basin of attraction of $y$.
\end{proof}

The size of the strong basins of attraction grows therefore with the number of low Notch
whose neighbouring high-Notch cells have other neighbours with low Notch.
For example, for the linear graphs $\P_L$ the size of the strong basin of attraction
is the largest for ``regular'' patterns, i.e., patterns that do not admit two adjacent cells with high Notch.

\begin{example}
  If $\G=\P_3$, the strong basin of attraction of $p_1=101010$ is given by the fixed point itself,
  whereas the strong basin of attraction of $p_2=010101$ is $J=\star10\star01\cup 01\star10\star$.
  The basin of attraction of $p_1$ is the set $\B^{6} \setminus J$,
  whereas the basin of attraction of $p_2$ is the set $\B^{6} \setminus \{p_1\}$ (see~\cref{fig:trapspaces} right).
\end{example}

\subsection{Summary and considerations on robustness of patterns}\label{sec:robustness}

We can use the characterisation of strong and weak basins of attraction to study
the robustness of stable patterns in response to small perturbations.
We want to answer the following questions:
\begin{enumerate}
  \item Which patterns can be obtained after perturbing a given pattern?
  \item Which perturbations do not affect the pattern?
  \item Can the system enter a cyclic path?
\end{enumerate}
The results of the previous section provide answers to these questions.
Consider a fixed point $x$, and call $y$ the state obtained by ``perturbing'' the pattern $x$.
Then, for the Boolean Delta-Notch model $F$, we have:
\begin{enumerate}
  \item the patterns that can be reached from $y$ are all the fixed points found in the minimal
        trap space $\k(y)$ containing $y$ (\cref{thm:all-fixed}),
  \item the system reaches exclusively the pattern $x$ if and only if $\k(y)$ can be written as $x[I]$
        with $\S(i)\cap I=\emptyset$ for all $i\in I$ (\cref{prop:strong_basins}), and
  \item in any other case, there are cyclic paths reachable from $y$ (\cref{prop:path01tox_pathyto01_cycle}~$(iv)$).
\end{enumerate}
On the other hand, for the reduced models $N$, while the result on the strong basins still holds (\cref{prop:strong-basins-N}),
not all fixed points contained in the minimal trap space are reachable (\cref{thm:reach-N}),
and cyclic paths are excluded (see~\cref{sec:asymptotic}).

\cref{prop:robustnessNany,prop:robustnessFany} show that, for both the one and two-variable model,
perturbations to a pattern do not propagate beyond cells at distance $2$.
The following result is a corollary:

\begin{proposition}\label{prop:robustness}
  Consider $x\in\B^{2L}$ fixed point for a Boolean Delta-Notch system, and take $i\in\cs$.
  \begin{enumerate}[label=(\roman*)]
    \item If $x_i=0$, then there exists a trap space $x[I\cup(I+L)]$ such that
      $\{i\}\subseteq I\subseteq\{i\}\cup\S(i)$.
    \item If $x_i=1$, then there exists a trap space $x[I\cup(I+L)]$ such that
      $\{i\}\subseteq I\subseteq\{i\}\cup\S(i)\cup\S(\S(i))$.
  \end{enumerate}
\end{proposition}

The analogous statement holds for $N$.
For changes of only one variable level in one cell, we have that:
\begin{itemize}
  \item \emph{Isolated changes of low Notch to high Notch, or high Delta to low Delta}
    can only affect direct neighbour cells.
  \item \emph{Isolated changes from high Notch to low Notch, or low Delta to high Delta}
    can only affect cells at maximum distance of $2$ from cell $i$.
\end{itemize}
The examples in~\cref{fig:perturbations} show that the bounds on the distance of affected cells are the smallest possible.
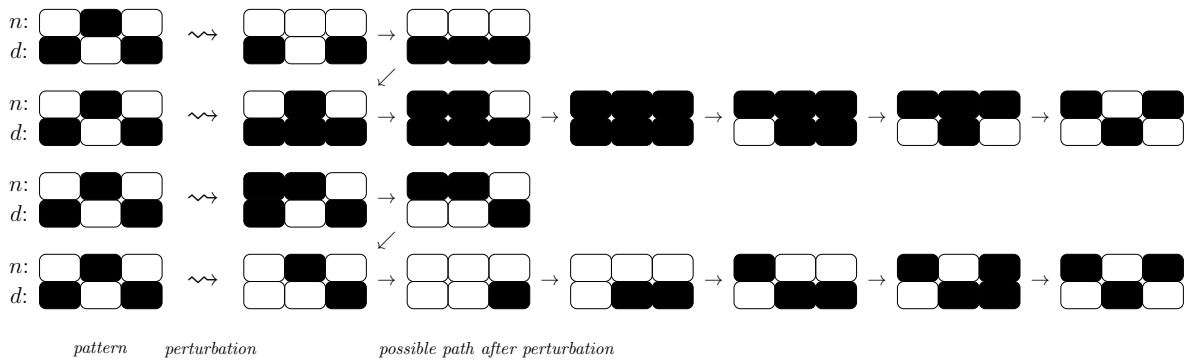
\begin{figure}
  \centering
  \resizebox{16.0cm}{!}{
  \newcommand*{\bs}{0.9}
  \newcommand*{\hg}{0.6}
  \newcommand*{\rd}{2.0}
  \begin{tikzpicture}[rounded corners]
    \node at (-\bs*0.5,\hg*1.5) {\Large $n$:};
    \node at (-\bs*0.5,\hg*0.5) {\Large $d$:};
    \node at (-\bs*0.5,-\hg*1.5) {\Large $n$:};
    \node at (-\bs*0.5,-\hg*2.5) {\Large $d$:};
    \draw[rounded corners] (0,\hg) rectangle (\bs,2*\hg);
    \draw[rounded corners,fill=black] (\bs,\hg) rectangle (2*\bs,2*\hg);
    \draw[rounded corners] (2*\bs,\hg) rectangle (3*\bs,2*\hg);
    \draw[rounded corners,fill=black] (0,0) rectangle (\bs,\hg);
    \draw[rounded corners] (\bs,0) rectangle (2*\bs,\hg);
    \draw[rounded corners,fill=black] (2*\bs,0) rectangle (3*\bs,\hg);
    \node at (\bs*4.0,\hg) {\huge $\leadsto$};
    \draw[rounded corners] (5*\bs,\hg) rectangle (6*\bs,2*\hg);
    \draw[rounded corners] (6*\bs,\hg) rectangle (7*\bs,2*\hg);
    \draw[rounded corners] (7*\bs,\hg) rectangle (8*\bs,2*\hg);
    \draw[rounded corners,fill=black] (5*\bs,0) rectangle (6*\bs,\hg);
    \draw[rounded corners] (6*\bs,0) rectangle (7*\bs,\hg);
    \draw[rounded corners,fill=black] (7*\bs,0) rectangle (8*\bs,\hg);
    \node at (\bs*8.5,\hg) {\large $\rightarrow$};
    \draw[rounded corners] (9*\bs,\hg) rectangle (10*\bs,2*\hg);
    \draw[rounded corners] (10*\bs,\hg) rectangle (11*\bs,2*\hg);
    \draw[rounded corners] (11*\bs,\hg) rectangle (12*\bs,2*\hg);
    \draw[rounded corners,fill=black] (9*\bs,0) rectangle (10*\bs,\hg);
    \draw[rounded corners,fill=black] (10*\bs,0) rectangle (11*\bs,\hg);
    \draw[rounded corners,fill=black] (11*\bs,0) rectangle (12*\bs,\hg);
    \node at (\bs*8.5,-0.5*\hg) {\large $\swarrow$};

    \draw[rounded corners] (0,-2*\hg) rectangle (\bs,-1*\hg);
    \draw[rounded corners,fill=black] (\bs,-2*\hg) rectangle (2*\bs,-1*\hg);
    \draw[rounded corners] (2*\bs,-2*\hg) rectangle (3*\bs,-1*\hg);
    \draw[rounded corners,fill=black] (0,-3*\hg) rectangle (\bs,-2*\hg);
    \draw[rounded corners] (\bs,-3*\hg) rectangle (2*\bs,-2*\hg);
    \draw[rounded corners,fill=black] (2*\bs,-3*\hg) rectangle (3*\bs,-2*\hg);
    \node at (\bs*4.0,-2*\hg) {\huge $\leadsto$};
    \draw[rounded corners] (5*\bs,-2*\hg) rectangle (6*\bs,-1*\hg);
    \draw[rounded corners,fill=black] (6*\bs,-2*\hg) rectangle (7*\bs,-1*\hg);
    \draw[rounded corners] (7*\bs,-2*\hg) rectangle (8*\bs,-1*\hg);
    \draw[rounded corners,fill=black] (5*\bs,-3*\hg) rectangle (6*\bs,-2*\hg);
    \draw[rounded corners,fill=black] (6*\bs,-3*\hg) rectangle (7*\bs,-2*\hg);
    \draw[rounded corners,fill=black] (7*\bs,-3*\hg) rectangle (8*\bs,-2*\hg);
    \node at (\bs*8.5,-2*\hg) {\large $\rightarrow$};
    \draw[rounded corners,fill=black] (9*\bs,-2*\hg) rectangle (10*\bs,-1*\hg);
    \draw[rounded corners,fill=black] (10*\bs,-2*\hg) rectangle (11*\bs,-1*\hg);
    \draw[rounded corners] (11*\bs,-2*\hg) rectangle (12*\bs,-1*\hg);
    \draw[rounded corners,fill=black] (9*\bs,-3*\hg) rectangle (10*\bs,-2*\hg);
    \draw[rounded corners,fill=black] (10*\bs,-3*\hg) rectangle (11*\bs,-2*\hg);
    \draw[rounded corners,fill=black] (11*\bs,-3*\hg) rectangle (12*\bs,-2*\hg);
    \node at (\bs*12.5,-2*\hg) {\large $\rightarrow$};
    \draw[rounded corners,fill=black] (13*\bs,-2*\hg) rectangle (14*\bs,-\hg);
    \draw[rounded corners,fill=black] (14*\bs,-2*\hg) rectangle (15*\bs,-\hg);
    \draw[rounded corners,fill=black] (15*\bs,-2*\hg) rectangle (16*\bs,-\hg);
    \draw[rounded corners,fill=black] (13*\bs,-3*\hg) rectangle (14*\bs,-2*\hg);
    \draw[rounded corners,fill=black] (14*\bs,-3*\hg) rectangle (15*\bs,-2*\hg);
    \draw[rounded corners,fill=black] (15*\bs,-3*\hg) rectangle (16*\bs,-2*\hg);
    \node at (\bs*16.5,-2*\hg) {\large $\rightarrow$};
    \draw[rounded corners,fill=black] (17*\bs,-2*\hg) rectangle (18*\bs,-\hg);
    \draw[rounded corners,fill=black] (18*\bs,-2*\hg) rectangle (19*\bs,-\hg);
    \draw[rounded corners,fill=black] (19*\bs,-2*\hg) rectangle (20*\bs,-\hg);
    \draw[rounded corners] (17*\bs,-3*\hg) rectangle (18*\bs,-2*\hg);
    \draw[rounded corners,fill=black] (18*\bs,-3*\hg) rectangle (19*\bs,-2*\hg);
    \draw[rounded corners,fill=black] (19*\bs,-3*\hg) rectangle (20*\bs,-2*\hg);
    \node at (\bs*20.5,-2*\hg) {\large $\rightarrow$};
    \draw[rounded corners,fill=black] (21*\bs,-2*\hg) rectangle (22*\bs,-\hg);
    \draw[rounded corners,fill=black] (22*\bs,-2*\hg) rectangle (23*\bs,-\hg);
    \draw[rounded corners,fill=black] (23*\bs,-2*\hg) rectangle (24*\bs,-\hg);
    \draw[rounded corners] (21*\bs,-3*\hg) rectangle (22*\bs,-2*\hg);
    \draw[rounded corners,fill=black] (22*\bs,-3*\hg) rectangle (23*\bs,-2*\hg);
    \draw[rounded corners] (23*\bs,-3*\hg) rectangle (24*\bs,-2*\hg);
    \node at (\bs*24.5,-2*\hg) {\large $\rightarrow$};
    \draw[rounded corners,fill=black] (25*\bs,-2*\hg) rectangle (26*\bs,-\hg);
    \draw[rounded corners] (26*\bs,-2*\hg) rectangle (27*\bs,-\hg);
    \draw[rounded corners,fill=black] (27*\bs,-2*\hg) rectangle (28*\bs,-\hg);
    \draw[rounded corners] (25*\bs,-3*\hg) rectangle (26*\bs,-2*\hg);
    \draw[rounded corners,fill=black] (26*\bs,-3*\hg) rectangle (27*\bs,-2*\hg);
    \draw[rounded corners] (27*\bs,-3*\hg) rectangle (28*\bs,-2*\hg);

    \node at (-\bs*0.5,-\hg*4.5) {\Large $n$:};
    \node at (-\bs*0.5,-\hg*5.5) {\Large $d$:};
    \node at (-\bs*0.5,-\hg*7.5) {\Large $n$:};
    \node at (-\bs*0.5,-\hg*8.5) {\Large $d$:};
    \draw[rounded corners,fill=black] (0,-6*\hg) rectangle (\bs,-5*\hg);
    \draw[rounded corners] (\bs,-6*\hg) rectangle (2*\bs,-5*\hg);
    \draw[rounded corners,fill=black] (2*\bs,-6*\hg) rectangle (3*\bs,-5*\hg);
    \draw[rounded corners] (0,-5*\hg) rectangle (\bs,-4*\hg);
    \draw[rounded corners,fill=black] (\bs,-5*\hg) rectangle (2*\bs,-4*\hg);
    \draw[rounded corners] (2*\bs,-5*\hg) rectangle (3*\bs,-4*\hg);
    \node at (\bs*4.0,-5*\hg) {\huge $\leadsto$};
    \draw[rounded corners,fill=black] (5*\bs,-6*\hg) rectangle (6*\bs,-5*\hg);
    \draw[rounded corners] (6*\bs,-6*\hg) rectangle (7*\bs,-5*\hg);
    \draw[rounded corners,fill=black] (7*\bs,-6*\hg) rectangle (8*\bs,-5*\hg);
    \draw[rounded corners,fill=black] (5*\bs,-5*\hg) rectangle (6*\bs,-4*\hg);
    \draw[rounded corners,fill=black] (6*\bs,-5*\hg) rectangle (7*\bs,-4*\hg);
    \draw[rounded corners] (7*\bs,-5*\hg) rectangle (8*\bs,-4*\hg);
    \node at (\bs*8.5,-5*\hg) {\large $\rightarrow$};
    \draw[rounded corners] (9*\bs,-6*\hg) rectangle (10*\bs,-5*\hg);
    \draw[rounded corners] (10*\bs,-6*\hg) rectangle (11*\bs,-5*\hg);
    \draw[rounded corners,fill=black] (11*\bs,-6*\hg) rectangle (12*\bs,-5*\hg);
    \draw[rounded corners,fill=black] (9*\bs,-5*\hg) rectangle (10*\bs,-4*\hg);
    \draw[rounded corners,fill=black] (10*\bs,-5*\hg) rectangle (11*\bs,-4*\hg);
    \draw[rounded corners] (11*\bs,-5*\hg) rectangle (12*\bs,-4*\hg);
    \node at (\bs*8.5,-6.5*\hg) {\large $\swarrow$};

    \draw[rounded corners,fill=black] (0,-9*\hg) rectangle (\bs,-8*\hg);
    \draw[rounded corners] (\bs,-9*\hg) rectangle (2*\bs,-8*\hg);
    \draw[rounded corners,fill=black] (2*\bs,-9*\hg) rectangle (3*\bs,-8*\hg);
    \draw[rounded corners] (0,-8*\hg) rectangle (\bs,-7*\hg);
    \draw[rounded corners,fill=black] (\bs,-8*\hg) rectangle (2*\bs,-7*\hg);
    \draw[rounded corners] (2*\bs,-8*\hg) rectangle (3*\bs,-7*\hg);
    \node at (\bs*4.0,-8*\hg) {\huge $\leadsto$};
    \draw[rounded corners] (5*\bs,-9*\hg) rectangle (6*\bs,-8*\hg);
    \draw[rounded corners] (6*\bs,-9*\hg) rectangle (7*\bs,-8*\hg);
    \draw[rounded corners,fill=black] (7*\bs,-9*\hg) rectangle (8*\bs,-8*\hg);
    \draw[rounded corners] (5*\bs,-8*\hg) rectangle (6*\bs,-7*\hg);
    \draw[rounded corners,fill=black] (6*\bs,-8*\hg) rectangle (7*\bs,-7*\hg);
    \draw[rounded corners] (7*\bs,-8*\hg) rectangle (8*\bs,-7*\hg);
    \node at (\bs*8.5,-8*\hg) {\large $\rightarrow$};
    \draw[rounded corners] (9*\bs,-9*\hg) rectangle (10*\bs,-8*\hg);
    \draw[rounded corners] (10*\bs,-9*\hg) rectangle (11*\bs,-8*\hg);
    \draw[rounded corners,fill=black] (11*\bs,-9*\hg) rectangle (12*\bs,-8*\hg);
    \draw[rounded corners] (9*\bs,-8*\hg) rectangle (10*\bs,-7*\hg);
    \draw[rounded corners] (10*\bs,-8*\hg) rectangle (11*\bs,-7*\hg);
    \draw[rounded corners] (11*\bs,-8*\hg) rectangle (12*\bs,-7*\hg);
    \node at (\bs*12.5,-8*\hg) {\large $\rightarrow$};
    \draw[rounded corners] (13*\bs,-9*\hg) rectangle (14*\bs,-8*\hg);
    \draw[rounded corners,fill=black] (14*\bs,-9*\hg) rectangle (15*\bs,-8*\hg);
    \draw[rounded corners,fill=black] (15*\bs,-9*\hg) rectangle (16*\bs,-8*\hg);
    \draw[rounded corners] (13*\bs,-8*\hg) rectangle (14*\bs,-7*\hg);
    \draw[rounded corners] (14*\bs,-8*\hg) rectangle (15*\bs,-7*\hg);
    \draw[rounded corners] (15*\bs,-8*\hg) rectangle (16*\bs,-7*\hg);
    \node at (\bs*16.5,-8*\hg) {\large $\rightarrow$};
    \draw[rounded corners] (17*\bs,-9*\hg) rectangle (18*\bs,-8*\hg);
    \draw[rounded corners,fill=black] (18*\bs,-9*\hg) rectangle (19*\bs,-8*\hg);
    \draw[rounded corners,fill=black] (19*\bs,-9*\hg) rectangle (20*\bs,-8*\hg);
    \draw[rounded corners,fill=black] (17*\bs,-8*\hg) rectangle (18*\bs,-7*\hg);
    \draw[rounded corners] (18*\bs,-8*\hg) rectangle (19*\bs,-7*\hg);
    \draw[rounded corners] (19*\bs,-8*\hg) rectangle (20*\bs,-7*\hg);
    \node at (\bs*20.5,-8*\hg) {\large $\rightarrow$};
    \draw[rounded corners] (21*\bs,-9*\hg) rectangle (22*\bs,-8*\hg);
    \draw[rounded corners,fill=black] (22*\bs,-9*\hg) rectangle (23*\bs,-8*\hg);
    \draw[rounded corners,fill=black] (23*\bs,-9*\hg) rectangle (24*\bs,-8*\hg);
    \draw[rounded corners,fill=black] (21*\bs,-8*\hg) rectangle (22*\bs,-7*\hg);
    \draw[rounded corners] (22*\bs,-8*\hg) rectangle (23*\bs,-7*\hg);
    \draw[rounded corners,fill=black] (23*\bs,-8*\hg) rectangle (24*\bs,-7*\hg);
    \node at (\bs*24.5,-8*\hg) {\large $\rightarrow$};
    \draw[rounded corners] (25*\bs,-9*\hg) rectangle (26*\bs,-8*\hg);
    \draw[rounded corners,fill=black] (26*\bs,-9*\hg) rectangle (27*\bs,-8*\hg);
    \draw[rounded corners] (27*\bs,-9*\hg) rectangle (28*\bs,-8*\hg);
    \draw[rounded corners,fill=black] (25*\bs,-8*\hg) rectangle (26*\bs,-7*\hg);
    \draw[rounded corners] (26*\bs,-8*\hg) rectangle (27*\bs,-7*\hg);
    \draw[rounded corners,fill=black] (27*\bs,-8*\hg) rectangle (28*\bs,-7*\hg);

    \node at (\bs*1.5,-\hg*0.5-10*\hg) {\emph{pattern}};
    \node at (\bs*4.2,-\hg*0.5-10*\hg) {\emph{perturbation}};
    \node at (\bs*11.2,-\hg*0.5-10*\hg) {\emph{possible path after perturbation}};
  \end{tikzpicture}}
  \caption{Changes in levels of Notch or Delta in one cell can induce the system to attain a different pattern.
    Changes to low levels of Notch or high levels of Delta can propagate to neighbour cells,
    and changes to high levels of Notch or low levels of Delta can affect cells at distance two
    (see~\cref{prop:robustness}). White represents high activity.}\label{fig:perturbations}
\end{figure}

\section{A generalisation}\label{sec:generalisation}

In this section we give a brief look at a class of networks that generalise
the models previously considered in this paper.
We fix again an undirected graph $\G$ without loops with vertex set $\cs$.
Given $k\in\mathbb{N}$, $k\geq 1$, consider the Boolean function $F^k\colon\B^{2L}\to\B^{2L}$
defined by
\begin{equation*}
  \begin{aligned}
    F^k_i(n,d)&=1 \text{ if and only if } \sum_{j\in\S(i)}d_j\geq k,\\
    F^k_{i+L}(n,d)&=\bar{n}^i,
  \end{aligned}
\end{equation*}
for all $i\in\cs$.
That is, at least $k$ high level of neighbouring Delta are required to activate Notch.
For $k=1$ we obtain the Delta-Notch model defined in~\cref{def:dn}.

We denote by $N^k\colon\B^L\to\B^L$ the reduced model
\begin{equation}\label{eq:dn-reduced-at-least-k}
  N^k_i(n) = 1 \text{ if and only if } \sum_{j\in\S(i)}\bar{n}_j\geq k\ \text{ for } i\in\cs.
\end{equation}

As seen in~\cref{sec:asymptotic} for $N$,
the network $N^k$ is a strict threshold network, with $A\in\{0,-1\}^{L\times L}$ and $b\in\mathbb{R}^L$ defined as follows:
\begin{equation*}
  A_{ij}=\begin{cases}-1 & \text{ if }j\in\S(i),\\0&\text{ otherwise,}\end{cases}
  \hspace{20pt}
  b_i=-|\S(i)|+k-\frac{1}{2}.
\end{equation*}

Since $A$ is symmetric and $A_{ii}\geq 0$ for all $i\in\cs$,
all the attractors for $AD_{N^k}$ are fixed points (\cite{goles1985decreasing}), and $AD_{N^k}$ has no cyclic paths.
By~\cref{thm:reduction}~$(i)$ the fixed points of $N$ and $F$ are
in one-to-one correspondence.
It was shown in~\cite{veliz2012computation} that the fixed points of $N$
are in one-to-one correspondence with the minimal vertex covers of the graph $\G$.
We show how this result can be partially extended to $N^k$.

In the following, we write $\P(A)$ for the subsets of a set $A$ and $\P_k(A)$ for the subsets of $A$ of size $k$.
Define the undirected hypergraph $\H(k)$ with vertex set $C$ and edge set
\begin{equation*}
  \begin{aligned}
  \{\{i\}\cup H\ | \ & i\in C, H\in\P_k(\S(i))\}.
  \end{aligned}
\end{equation*}

The edges of $\H(k)$ are given by subsets of the vertices $\cs$
of cardinality $k+1$, each consisting of a vertex and $k$ of its neighbours.

Recall that a \emph{transversal} or \emph{hitting set} of a hypergraph is a subset of the vertices
that has non-empty intersection with every edge.
We introduce the following terminology: we say that a transversal $Q$ of $\H(k)$ is \emph{$k$-minimal} if, for each $i\in Q$, $|\S(i)\cap Q|\leq |\S(i)|-k$.
Note that a $k$-minimal transversal does not contain any vertex with fewer than $k$ neighbours in $\G$.

\begin{theorem}\label{thm:ss-vc-hyperg-2}
  The fixed points for $N^k$ and $F^k$
  are in one-to-one correspondence with the $k$-minimal transversals of the hypergraph $\H(k)$.
\end{theorem}
\begin{proof}
  Consider the bijective map $h\colon\B^L\to\mathcal{P}(\cs)$ defined by $x\mapsto\{i\in\cs\ | \ x_i=1\}$,
  and let $n\in\B^L$ be a fixed point of $N^k$. Observe that $n_j=0$ for all $j$ such that $|\S(j)|<k$.
  Take $I$ edge in $\H(k)$, and suppose that $i\in I$ and $H\in\P_k(\S(i))$ are such that $I=\{i\}\cup H$.
  Since $n_{i}=\bigvee_{J\in\P_k(\S(i))}\bigwedge_{j\in J}\bar{n}_j$, either $n_{i}=1$ or $n_j=1$ for some $j\in H$.
  Hence $h(n)$ is a transversal.

  To see that $h(n)$ is $k$-minimal, take $i\in h(n)$.
  Since $n_i=1$, there exists a subset $H\in\P_k(\S(i))$ such that $n_j=0$ and $j\notin h(n)$ for all $j\in H$.
  Hence $|\S(i)\cap h(n)|\leq |\S(i)|-|H|=|\S(i)|-k$.

  Vice versa, consider a $k$-minimal transversal $Q$ of $\H(k)$, and define $n=h^{-1}(Q)$.
  Given $i\in\cs$, if $\sum_{j\in\S(i)} \bar{n}_j\geq k$, then there exists $H\in\P_k(\S(i))$
  such that $n_j=0$ and $j\notin Q$ for all $j\in H$.
  Hence $\{i\}\cup H$ is an edge in $\H(k)$ and since $Q$ is a transversal we must have $i\in Q$ and $n_i=1$.
  If instead $\sum_{j\in\S(i)} \bar{n}_j<k$, then $|\S(i)|-|\S(i)\cap Q|<k$, and
  since $Q$ is $k$-minimal, we find $i\notin Q$ and $n_i=0$.
\end{proof}

As in~\cref{thm:homogeneous-N}, it is possible to show that all fixed points are reachable
from homogeneous initial conditions.
We now give a description of the trap spaces for $N^k$ and $F^k$.

\begin{proposition}\label{prop:trap-spaces-Nk}
The trap spaces for $N^k$ are of the form $x[I]$, with $x$ fixed point, and
for all $i\in\S(I)\cap I^{\mathsf{c}}$:
\begin{itemize}
\item[(i)] if $x_i=1$, the set $\{j\in\S(i)\cap I^{\mathsf{c}}\ | \ x_i=0\}$
  has cardinality greater or equal to $k$;
\item[(ii)] if $x_i=0$, the set $\{j\in\S(i)\cap I^{\mathsf{c}}\ | \ x_i=0\}\cup(\S(i)\cap I)$
  has cardinality smaller than $k$.
\end{itemize}
\end{proposition}
\begin{proof}
  Consider a subspace $x[I]$ as in the statement, and take $y\in x[I]$.
  We need to show that all successors of $y$ in the asynchronous state transition graph are in $x[I]$,
  or, in other words, $N^k_i(y)=y_i$ for all $i\notin I$.
  If $\S(i)\cap I=\emptyset$, then the conclusion follows from the fact that $x$ is a fixed point.
  If $i\in\S(I)$, and $y_i=1$, then $N^k_i(y)=1$ follows from $(i)$,
  and if $y_i=0$, $N^k_i(y)=0$ follows from $(ii)$.

  Vice versa, consider a trap space $x[I]$.
  Since we must have $N^k_i(x)=x_i$ for all $i\notin I$, and all attractors of $N^k$ are fixed points,
  we can assume that $x$ is a fixed point.
  Take $i\in\S(I)\cap I^{\mathsf{c}}$ with $x_i=1$, and $y\in x[I]$ such that $y_j=1$ for all $j\in\S(i)\cap I$.
  Then $1=x_i=N^k_i(y)$ shows point $(i)$.
  If $i\in\S(I)\cap I^{\mathsf{c}}$ is such that $x_i=0$, taking $y\in x[I]$ such that $y_j=0$ for all $j\in\S(i)\cap I$
  gives point $(ii)$.
\end{proof}

\begin{proposition}\label{prop:trap-spaces-Fk}
The trap spaces for $F^k$ are of the form $x[I]$, with $x$ fixed point,
$I_N\subseteq I_D$, and, for $i\in I^{\mathsf{c}}_N$:
\begin{itemize}
\item[(i)] if $x_i=1$, the set $\{j\in\S(i)\cap I^{\mathsf{c}}_D\ | \ x_{i+L}=1\}$
  has cardinality greater or equal to $k$;
\item[(ii)] if $x_i=0$, the set $\{j\in\S(i)\cap I^{\mathsf{c}}_D\ | \ x_{i+L}=1\}\cup(\S(i)\cap I_D)$
  has cardinality smaller than $k$.
\end{itemize}
\end{proposition}
\begin{proof}
  Consider a subspace $x[I]$ as in the statement, and take $y\in x[I]$.
  Then for $i\in I^{\mathsf{c}}_N$ we have $y_i=x_i$, and in both cases we have $F^k_i(y)=F^k_i(x)=x_i$.
  For $i\in I^{\mathsf{c}}_D$, $y_{i+L}=x_{i+L}$ and since $x$ is fixed, $F^k_{i+L}(y)=F^k_{i+L}(x)=y_{i+L}$.

  Vice versa, consider a trap space $x[I]$.
  The containment $I_N\subseteq I_D$ follows from the definition of $F$.
  Since we must have $F^k_i(x)=x_i$ for all $i\notin I$, and all attractors of $F^k$ are fixed points,
  we can assume that $x$ is a fixed point.
  Take $i\in\S(I)\cap I^{\mathsf{c}}_N$ with $x_i=1$, and $y\in x[I]$ such that $y_{j+L}=0$ for all $j\in\S(i)\cap I_D$.
  Then $1=x_i=F^k_i(y)$ shows point $(i)$.
  If $i\in\S(I)\cap I^{\mathsf{c}}_N$ is such that $x_i=0$, taking $y\in x[I]$ such that $y_{j+L}=1$ for all $j\in\S(i)\cap I_D$
  gives point $(ii)$.
\end{proof}

Recall that for the case $k=1$ we were able to describe the minimal trap space containing
a fixed point and some of its adjacent states (\cref{prop:robustnessNany,prop:robustnessFany}),
and to show that changes in a pattern can not propagate to cells at distance greater than $2$.
The following example shows that a similar result does not hold for $k>1$.
The characterisations of the basins of attraction for $N$ and $F$ also do not immediately
generalise to $N^k$ and $F^k$, and are left as open problems.

\begin{example}\label{ex:robNk}
  For $N^k$ (and $F^k$) with $k=2$, one can construct a network such that
  a change in one cell can cause repercussions at arbitrary distance.
  Consider the example in~\cref{fig:robNk} left.
  By changing the low level (in black) to high level (in white) in the cell with a dashed border,
  the pattern on the right can be reached.
  The network can be made as large as wanted.
  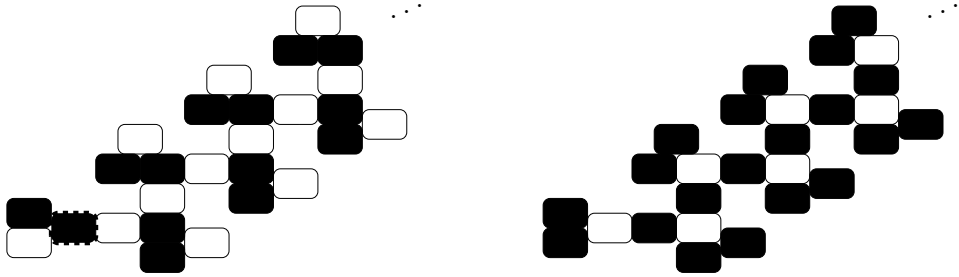
\begin{figure}
    \centering
    \begin{minipage}{7.0cm}
      \resizebox{6.0cm}{!}{
      \newcommand*{\bs}{0.9}
      \newcommand*{\hg}{0.6}
      \newcommand*{\rd}{2.0}
      \begin{tikzpicture}[rounded corners]
        \draw[rounded corners,fill=black] (-\bs,0.5*\hg) rectangle (0,1.5*\hg);
        \draw[rounded corners] (-\bs,-0.5*\hg) rectangle (0,0.5*\hg);
        \draw[rounded corners,fill=black,dashed,line width=1mm] (0,0) rectangle (\bs,\hg);

        \draw[rounded corners] (\bs,0) rectangle (2*\bs,\hg);
        \draw[rounded corners,fill=black] (2*\bs,0) rectangle (3*\bs,\hg);
        \draw[rounded corners] (3*\bs,-0.5*\hg) rectangle (4*\bs,0.5*\hg);
        \draw[rounded corners,fill=black] (2*\bs,0) rectangle (3*\bs,-\hg);
        \draw[rounded corners] (2*\bs,\hg) rectangle (3*\bs,2*\hg);
        \draw[rounded corners,fill=black] (2*\bs,2*\hg) rectangle (3*\bs,3*\hg);
        \draw[rounded corners] (1.5*\bs,3*\hg) rectangle (2.5*\bs,4*\hg);
        \draw[rounded corners,fill=black] (\bs,2*\hg) rectangle (2*\bs,3*\hg);

        \draw[rounded corners] (3*\bs,2*\hg) rectangle (4*\bs,3*\hg);
        \draw[rounded corners,fill=black] (4*\bs,2*\hg) rectangle (5*\bs,3*\hg);
        \draw[rounded corners] (5*\bs,1.5*\hg) rectangle (6*\bs,2.5*\hg);
        \draw[rounded corners,fill=black] (4*\bs,2*\hg) rectangle (5*\bs,1*\hg);
        \draw[rounded corners] (4*\bs,3*\hg) rectangle (5*\bs,4*\hg);
        \draw[rounded corners,fill=black] (4*\bs,4*\hg) rectangle (5*\bs,5*\hg);
        \draw[rounded corners] (3.5*\bs,5*\hg) rectangle (4.5*\bs,6*\hg);
        \draw[rounded corners,fill=black] (3*\bs,4*\hg) rectangle (4*\bs,5*\hg);

        \draw[rounded corners] (5*\bs,4*\hg) rectangle (6*\bs,5*\hg);
        \draw[rounded corners,fill=black] (6*\bs,4*\hg) rectangle (7*\bs,5*\hg);
        \draw[rounded corners] (7*\bs,3.5*\hg) rectangle (8*\bs,4.5*\hg);
        \draw[rounded corners,fill=black] (6*\bs,4*\hg) rectangle (7*\bs,3*\hg);
        \draw[rounded corners] (6*\bs,5*\hg) rectangle (7*\bs,6*\hg);
        \draw[rounded corners,fill=black] (6*\bs,6*\hg) rectangle (7*\bs,7*\hg);
        \draw[rounded corners] (5.5*\bs,7*\hg) rectangle (6.5*\bs,8*\hg);
        \draw[rounded corners,fill=black] (5*\bs,6*\hg) rectangle (6*\bs,7*\hg);

        \node at (8.0*\bs,8.0*\hg) {\huge \reflectbox{$\ddots$}};
      \end{tikzpicture}}
    \end{minipage}
    \begin{minipage}{6.0cm}
      \resizebox{6.0cm}{!}{
      \newcommand*{\bs}{0.9}
      \newcommand*{\hg}{0.6}
      \newcommand*{\rd}{2.0}
      \begin{tikzpicture}[rounded corners]
        \draw[rounded corners,fill=black] (-\bs,0.5*\hg) rectangle (0,1.5*\hg);
        \draw[rounded corners,fill=black] (-\bs,-0.5*\hg) rectangle (0,0.5*\hg);
        \draw[rounded corners] (0,0) rectangle (\bs,\hg);

        \draw[rounded corners,fill=black] (\bs,0) rectangle (2*\bs,\hg);
        \draw[rounded corners] (2*\bs,0) rectangle (3*\bs,\hg);
        \draw[rounded corners,fill=black] (3*\bs,-0.5*\hg) rectangle (4*\bs,0.5*\hg);
        \draw[rounded corners,fill=black] (2*\bs,0) rectangle (3*\bs,-\hg);
        \draw[rounded corners,fill=black] (2*\bs,\hg) rectangle (3*\bs,2*\hg);
        \draw[rounded corners] (2*\bs,2*\hg) rectangle (3*\bs,3*\hg);
        \draw[rounded corners,fill=black] (1.5*\bs,3*\hg) rectangle (2.5*\bs,4*\hg);
        \draw[rounded corners,fill=black] (\bs,2*\hg) rectangle (2*\bs,3*\hg);

        \draw[rounded corners,fill=black] (3*\bs,2*\hg) rectangle (4*\bs,3*\hg);
        \draw[rounded corners] (4*\bs,2*\hg) rectangle (5*\bs,3*\hg);
        \draw[rounded corners,fill=black] (5*\bs,1.5*\hg) rectangle (6*\bs,2.5*\hg);
        \draw[rounded corners,fill=black] (4*\bs,2*\hg) rectangle (5*\bs,1*\hg);
        \draw[rounded corners,fill=black] (4*\bs,3*\hg) rectangle (5*\bs,4*\hg);
        \draw[rounded corners] (4*\bs,4*\hg) rectangle (5*\bs,5*\hg);
        \draw[rounded corners,fill=black] (3.5*\bs,5*\hg) rectangle (4.5*\bs,6*\hg);
        \draw[rounded corners,fill=black] (3*\bs,4*\hg) rectangle (4*\bs,5*\hg);

        \draw[rounded corners,fill=black] (5*\bs,4*\hg) rectangle (6*\bs,5*\hg);
        \draw[rounded corners] (6*\bs,4*\hg) rectangle (7*\bs,5*\hg);
        \draw[rounded corners,fill=black] (7*\bs,3.5*\hg) rectangle (8*\bs,4.5*\hg);
        \draw[rounded corners,fill=black] (6*\bs,4*\hg) rectangle (7*\bs,3*\hg);
        \draw[rounded corners,fill=black] (6*\bs,5*\hg) rectangle (7*\bs,6*\hg);
        \draw[rounded corners] (6*\bs,6*\hg) rectangle (7*\bs,7*\hg);
        \draw[rounded corners,fill=black] (5.5*\bs,7*\hg) rectangle (6.5*\bs,8*\hg);
        \draw[rounded corners,fill=black] (5*\bs,6*\hg) rectangle (6*\bs,7*\hg);

        \node at (8.0*\bs,8.0*\hg) {\huge \reflectbox{$\ddots$}};
      \end{tikzpicture}}
    \end{minipage}
    \caption{Example showing the propagation of a pattern perturbation in $AD_{N^k}$ for $k=2$.
      White cells have high levels of Notch.
      The pattern on the right can be reached from the state obtained from the
      pattern on the left when changing the level of Notch in the cell
      with a dashed border.}\label{fig:robNk}
  \end{figure}
\end{example}

\section{Conclusion and prospects}\label{sec:conclusion}

In this work we gave some characterisations of the dynamics of simple Boolean models
of the Delta-Notch system, complementing existing computationally-costly
algorithmic analyses (e.g.~\cite{mendes2013composition,varela2018stable}).
We considered models with two variables per cell, and reduced models with only one
variable per cell.
Results on Boolean threshold networks~\cite{goles1985decreasing} imply that all attractors are fixed points,
and that the asynchronous dynamics of reduced models do not contain any cyclic path.
In addition, the identification of the fixed points can be traced back to determining the minimal vertex covers
(or the maximal independent vertex sets) of the graph representing the neighbour relation between cells \cite{veliz2012computation}.
The emerging patterns are consistent with those obtained in the spatially-discrete continuous model of~\cite{collier1996pattern}.
We gave a characterisation of the trap spaces~(\cref{thm:trapspaces,thm:trap-spaces-N}) and of the patterns that can
be reached from a given state~(\cref{thm:all-fixed,thm:reach-N}) for both the one- and two-variable models.
In particular, we saw that all patterns can be obtained from homogeneous starting points~(\cref{thm:homogeneous,thm:homogeneous-N}).
For the two-variable models, all the fixed points in the minimal trap space
containing the initial state are reachable, a property that does not hold for the one-variable models.
The effects of cell perturbations on patterns were discussed in~\cref{sec:robustness}:
changes in patterns can only propagate to cells at maximum distance $2$.
Finally, we considered a generalisation of the models (\cref{sec:generalisation}), where Notch is assumed to be
activated when a certain minimum amount of neighbour cells with high levels of Delta is reached, as in~\cite{varela2018epilog}.
Although results on the asymptotic behaviour extend to these models, we showed with an example~(\ref{ex:robNk})
that the characterisation of the minimal trap spaces does not in general extend.
We leave as open question the problem of determining if some results on the reachability
and trap spaces can be extended to these models under some assumptions on the underlying graph.

Our results concern the structure of the dynamics and do not allow for quantitative results
regarding, for instance, the distribution of Notch obtained with trajectories
starting from a given initial condition, as considered, for example, in~\cite{varela2018epilog}.
The study of the asynchronous dynamics as a Markov chain is used to quantify simulation results of Boolean models (\cite{stoll2017maboss})
and could help with the interpretation of simulation results.
The model presented here provides a basis for the exploration of networks with more elaborate cell modules,
and for the investigation of the role of the simple mechanism we considered in the generation of
spatial inhomogeneity in more complex Boolean systems.

\section*{Acknowledgements}
The authors are grateful to C. Chaouiya and E. Remy for helpful discussions,
and to the reviewers for their useful comments.

\section*{Funding}
Funded by the Volkswagen Stiftung (Volkswagen Foundation) under the
funding initiative Life? - A fresh scientific approach to the basic
principles of life (project ID: 93063).

\bibliographystyle{plainnat}
\bibliography{biblio}

\end{document}